\documentclass[12pt]{article}
\usepackage{graphicx,amsfonts,amsmath,amssymb,amsthm,mathrsfs,url,hyperref,tikz,natbib}

\title{Gibbs and Boltzmann Entropy in Classical and Quantum Mechanics}
\author{
Sheldon Goldstein\footnote{Departments of Mathematics, Physics, and Philosophy, 
	Rutgers University, Hill Center, 110 Frelinghuysen Road, Piscataway, 
	NJ 08854-8019, USA. Email: oldstein@math.rutgers.edu},\ 
Joel L.~Lebowitz\footnote{Departments of Mathematics and Physics, 
	Rutgers University, Hill Center, 110 Frelinghuysen Road, Piscataway, 
	NJ 08854-8019, USA. Email: lebowitz@math.rutgers.edu},\ 
Roderich Tumulka\footnote{Mathematisches Institut,
     Eberhard-Karls-Universit\"at, Auf der Morgenstelle 10, 72076
     T\"ubingen, Germany. Email: roderich.tumulka@uni-tuebingen.de},\ and
Nino Zangh\`\i\footnote{Dipartimento di Fisica, Universit\`a di Genova, and
	Istituto Nazionale di Fisica Nucleare (Sezione di Genova),
	Via Dodecaneso 33, 16146 Genova, Italy. Email: zanghi@ge.infn.it}
}
\date{June 2, 2019}

\newcommand{\Hilbert}{\mathscr{H}}
\newcommand{\Kilbert}{\mathscr{K}}

\newcommand{\EEE}{\mathbb{E}}

\newcommand{\RRR}{\mathbb{R}}

\newcommand{\SSS}{\mathbb{S}}
\newcommand{\ZZZ}{\mathbb{Z}}

\newcommand{\pr}[1]{| #1 \rangle \langle #1 |}

\DeclareMathOperator{\tr}{tr}

\newcommand{\vq}{\boldsymbol{q}}

\newcommand{\vv}{\boldsymbol{v}}

\newcommand{\vomega}{\boldsymbol{\omega}}

\DeclareMathOperator{\vol}{\mathrm{vol}}
\newcommand{\kk}{k} 
\newcommand{\sys}{\mathscr{S}}
\newcommand{\can}{\mathrm{can}}
\newcommand{\mc}{\mathrm{mc}}

\newcommand{\eq}{\mathrm{eq}}

\newcommand{\G}{\mathrm{G}}
\newcommand{\B}{\mathrm{B}}
\newcommand{\vN}{\mathrm{vN}}
\newcommand{\qB}{\mathrm{qB}}
\newcommand{\X}{\mathscr{X}}
\newcommand{\Y}{\mathscr{Y}}

\theoremstyle{plain}
\newtheorem{thm}{Theorem}
\newtheorem{prop}{Proposition}
\newtheorem{conj}{Conjecture}
\theoremstyle{definition}

\newcommand{\be}{\begin{equation}}
\newcommand{\ee}{\end{equation}}

\begin{document}
\maketitle
\begin{abstract}
The Gibbs entropy of a macroscopic classical system is a function of a probability distribution over phase space, i.e., of an ensemble. In contrast, the Boltzmann entropy is a function on phase space, and is thus defined for an individual system. Our aim is to discuss and compare these two notions of entropy, along with the associated ensemblist and individualist views of thermal equilibrium. Using the Gibbsian ensembles for the computation of the Gibbs entropy, the two notions yield the same (leading order) values for the entropy of a macroscopic system in thermal equilibrium. The two approaches do not, however, necessarily agree for non-equilibrium systems. For those, we argue that the Boltzmann entropy is the one that corresponds to thermodynamic entropy, in particular in connection with the second law of thermodynamics. Moreover, we describe the quantum analog of the Boltzmann entropy, and we argue that the individualist (Boltzmannian) concept of equilibrium is supported by the recent works on thermalization of closed quantum systems.

\bigskip

\end{abstract}

\newpage

\tableofcontents

\section{Introduction}

Disagreement among scientists is often downplayed, and science often presented as an accumulation of discoveries, of universally accepted contributions to our common body of knowledge. But in fact, there is substantial disagreement among physicists, not only concerning questions that we have too little information about to settle them, such as the nature of dark matter, but also concerning conceptual questions about which all facts have long been in the literature, such as the interpretation of quantum mechanics. Another question of the latter type concerns the definition of entropy and some related concepts. In particular, two different formulations are often given in the literature for how to define the thermodynamic entropy (in equilibrium and non-equilibrium states) of a macroscopic physical system in terms of a microscopic, mechanical description (classical or quantum).

\subsection{Two Definitions of Entropy in Classical Statistical Mechanics}
\label{sec:twodef}

In classical mechanics, the \emph{Gibbs entropy} of a physical system with phase space $\X$, for example $\X = \RRR^{6N}=\{(\vq_1,\vv_1,\ldots,\vq_N,\vv_N)\}$ for $N$ point particles in $\RRR^3$ with positions $\vq_j$ and velocities $\vv_j$, is defined as
\be\label{SGdef}
S_\G(\rho) = - \kk \int_{\X} \!\! dx \, \rho(x) \, \log \rho(x)\,,
\ee
where $\kk$ is the Boltzmann constant, $dx=N!^{-1}d^3\vq_1 \, d^3\vv_1\cdots d^3\vq_N\, d^3\vv_N$ the (symmetrized) phase space volume measure, $\log$ the natural logarithm,\footnote{One actually takes the expression $u \log u$ to mean the continuous extension of the function $u \mapsto u\log u$ from $(0,\infty)$ to the domain $[0,\infty)$; put differently, we follow the convention to set $0\log 0 =0$.} and $\rho$ a probability density on $\X$.\footnote{Changing the unit of phase space volume will change $\rho(x)$ by a constant factor and thus $S_\G(\rho)$ by addition of a constant, an issue that does not matter for the applications and disappears anyway in the quantum case.}

The \emph{Boltzmann entropy} of a macroscopic system is defined as
\be\label{SBdef}
S_\B(X) = \kk \,\log \vol \Gamma(X)\,,
\ee
where $X\in\X$ is the actual phase point of the system, $\vol$ means the volume in $\X$, and $\Gamma(X)$ is the set of all phase points that ``look macroscopically the same'' as $X$. Obviously, there is no unique precise definition for ``looking macroscopically the same,'' so we have a certain freedom to make a reasonable choice. It can be argued that for large numbers $N$ of particles (as appropriate for macroscopic physical systems), the arbitrariness in the choice of $\Gamma(X)$ shrinks and becomes less relevant.\footnote{For example, for a dilute gas and large $N$, $-S_\B(X)/\kk +N$ equals approximately the $H$ functional, i.e., the integral in \eqref{-Hdef} below, which does not refer any more to a specific choice of boundaries of $\Gamma(X)$.} A convenient procedure is to partition the phase space into regions $\Gamma_\nu$ we call macro sets (see Figure~\ref{fig:phasespace}), 
\be\label{partition}
\X = \bigcup_\nu \Gamma_\nu\,,
\ee
and to take as $\Gamma(X)$ the $\Gamma_\nu$ containing $X$. That is,
\be
S_\B(X) = S_\B(\nu) := \kk \log \vol \Gamma_\nu\,.
\ee
We will give more detail in Section~\ref{sec:vision}. Boltzmann's definition \eqref{SBdef} is often abbreviated as ``$S_\B=\kk \log W$'' with $W=\vol \Gamma(X)$. In every energy shell there is usually one macro set $\Gamma_\nu= \Gamma_\eq$ that corresponds to thermal equilibrium and takes up by far most (say, more than $99.99\%$) of the volume (see Section~\ref{sec:macrostates}).

\begin{figure}[h]
\begin{center}
\includegraphics[width=7cm]{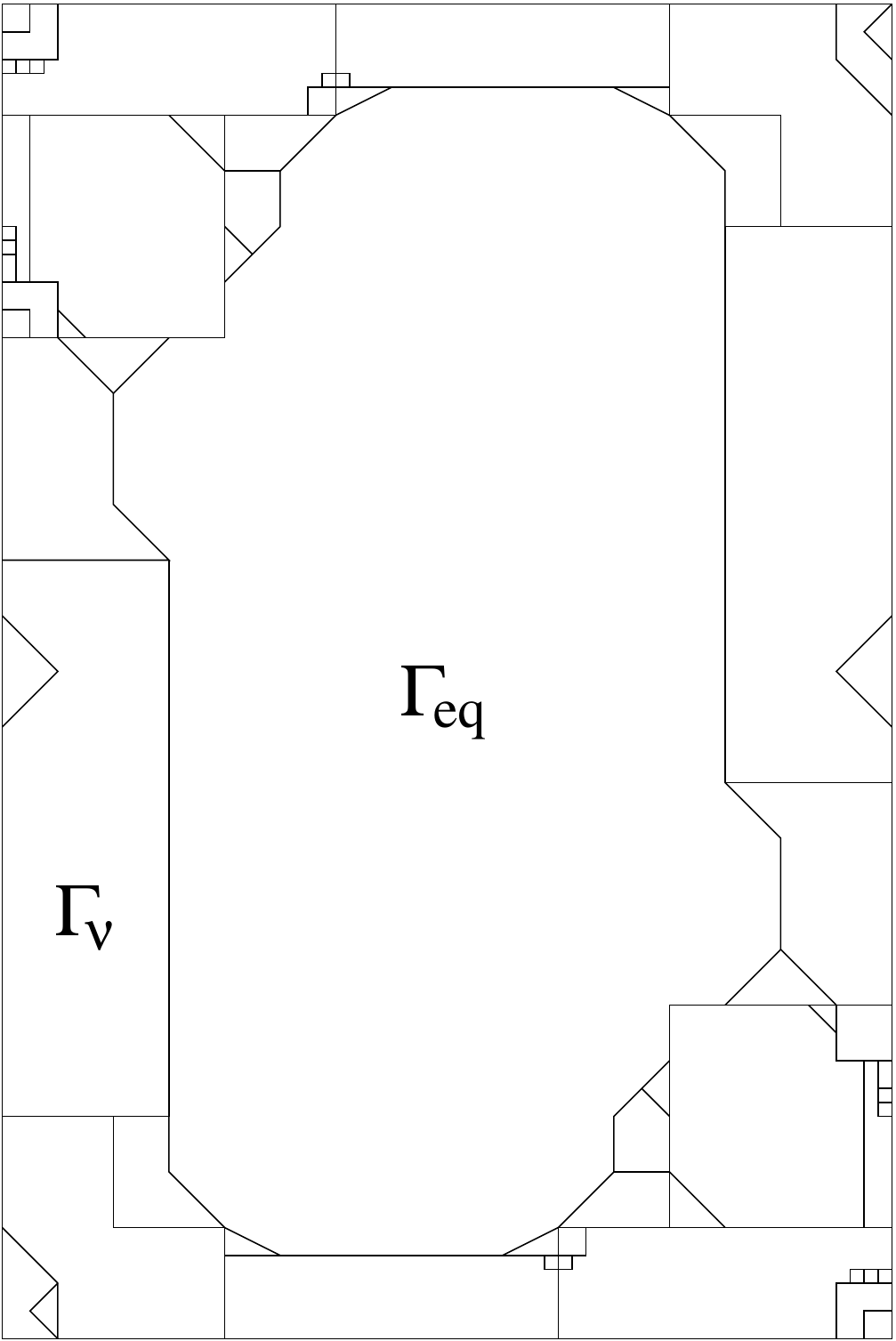}
\end{center}
\caption{Partition of phase space, or rather an energy shell therein, into macro sets $\Gamma_\nu$, with the thermal equilibrium set taking up most of the volume (not drawn to scale). Reprinted from \citep{GHLT17}.}
\label{fig:phasespace}
\end{figure}

\subsection{$X$ vs.\ $\rho$}
\label{sec:xrho}

An immediate problem with the Gibbs entropy is that while every classical system has a definite phase point $X$ (even if we observers do not know it), a system does not ``have a $\rho$''; that is, it is not clear which distribution $\rho$ to use. For a system in thermal equilibrium, $\rho$ presumably means a Gibbsian equilibrium ensemble (micro-canonical, canonical, or grand-canonical). It follows that, for thermal equilibrium states, $S_\B$ and $S_\G$ agree to leading order, see \eqref{SBeqSGmc} below. In general, several possibilities for $\rho$ come to mind:
\begin{itemize}
\item[(a)] ignorance: $\rho(x)$ expresses the strength of an observer's belief that $X=x$.
\item[(b)] preparation procedure: A given procedure does not always reproduce the same phase point, but produces a random phase point with distribution $\rho$.
\item[(c)] coarse graining: Associate with every $X\in\X$ a distribution $\rho_X(x)$ on $\X$ that captures how macro-similar $X$ and $x$ are (or perhaps, how strongly an ideal observer seeing a system with phase point $X$ would believe that $X=x$).
\end{itemize}
Correspondingly, there are several different notions of Gibbs entropy, which we will discuss in Sections~\ref{sec:subjective} and \ref{sec:agree}. Here, maybe (c) could be regarded as a special case of (a), and thermal equilibrium ensembles as a special case of (b). In fact, it seems that Gibbs himself had in mind that any system in thermal equilibrium has a random phase point whose distribution $\rho$ should be used, which is consistent with option (b); in his words \citep[p.~152]{Gib02}: 
\begin{quotation}
[\ldots] we shall find that [the] distinction [between interaction of a system $S_1$ with a system $S_2$ with determined phase point $X_2$ and one with distribution $\rho_2$] corresponds to the distinction in thermodynamics between mechanical and thermal action.
\end{quotation}
In our discussion we will also address the status of the Gibbsian ensembles \citep[see also][]{Gol19}. We will argue that $S_\B$ qualifies as a definition of thermodynamic entropy whereas version (a) of $S_\G$ does not; (b) is not correct in general; and (c) is acceptable to the extent that it is not regarded as a special case of (a). 

Different views about the meaning of entropy and the second law have consequences about the explanatory and predictive power of the second law that we will consider in Section~\ref{sec:subjective}. They also have practical consequences in the formulation of hydrodynamic equations, e.g., Navier-Stokes equations, for macroscopic variables \citep{GL04}: such macroscopic equations can be compatible with a microscopic Hamiltonian evolution only if they make sure that the Boltzmann entropy increases.

\bigskip

The remainder of this article is organized as follows. In Section~\ref{sec:status2nd}, we raise the question of the status of the second law for $S_\G$ and $S_\B$. In Section~\ref{sec:quantumcase}, we consider the analogs of Gibbs and Boltzmann entropy in quantum mechanics. The following Sections \ref{sec:subjective}--\ref{sec:indep} focus again on the classical case for simplicity. In Section~\ref{sec:subjective}, we discuss and criticize option (a), the idea that entropy is about subjective knowledge. In Section~\ref{sec:vision}, we explain why Boltzmann entropy indeed tends to increase with time and discuss doubts and objections to this statement. In Section~\ref{sec:agree}, we discuss an individualist understanding of Gibbs entropy as a generalization of Boltzmann entropy. In Section~\ref{sec:status}, we discuss the status of Gibbs's ensembles. In Section~\ref{sec:indep}, we comment on a few proposals for how entropy increase should work for Gibbs entropy. In Section~\ref{sec:qm}, we add some deeper considerations about entropy in quantum mechanics. In Section~\ref{sec:conclusions}, we conclude.

\section{Status of the Second Law}
\label{sec:status2nd}

\subsection{Gibbs Entropy}

Another immediate problem with the Gibbs entropy is that it does not change with time, 
\be\label{SGt}
\frac{dS_\G(\rho_t)}{dt}=0\,,
\ee
if $\rho$ is taken to evolve in accord with the microscopic Hamiltonian dynamics, that is, according to the Liouville equation
\be\label{Liouville}
\frac{\partial \rho}{\partial t}= - \sum_{i=1}^{6N} \frac{\partial}{\partial x_i} \bigl( \rho(x)\, v(x) \bigr)\,,
\ee
where $v(x)$ is the vector field on $\X$ that appears in the equation of motion
\be
\frac{dX_i}{dt} = v_i(X(t))
\ee
for the phase point $X(t)\in\X$, such as $v=\omega \nabla H$ with $6N\times 6N$ matrix $\omega= \bigl(\begin{smallmatrix} 0&I\\-I&0\end{smallmatrix} \bigr)$ in position and momentum coordinates and $H$ the Hamiltonian function. Generally, the Gibbs entropy does not change when $\rho$ gets transported by any volume-preserving bijection $\Phi:\X \to \X$, 
\be\label{SGPhi}
S_\G(\rho \circ \Phi^{-1}) = S_\G(\rho)\,.
\ee
In particular, by Liouville's theorem of the conservation of phase space volume, this applies when $\Phi=\Phi_t$ is the Hamiltonian time evolution, with $X(t) = \Phi_t(X(0))$. The time independence of $S_\G(\rho)$ conflicts with the formulation of the second law given by Clausius, the person who coined the ``laws of thermodynamics'' as follows \citep[p.~365]{Cla65}:
\begin{quotation}
    1. The energy of the universe is constant.

    2. The entropy of the universe tends to a maximum.
\end{quotation}
Among the authors who took entropy to be the Gibbs entropy, some \cite[e.g.,][]{Khi} have argued that Clausius's wording of the second law is inappropriate or exaggerated, others \citep[e.g.,][]{Mac89} that the Liouville evolution \eqref{Liouville} is not the relevant evolution here. We will come back to this point in Section~\ref{sec:indep}. As we will explain in Section~\ref{sec:vision}, Clausius's statement is actually correct for the Boltzmann entropy.

\subsection{Boltzmann's $H$ and $\kk \log W$}

We should address right away a certain confusion about Gibbs and Boltzmann entropy that has come from the fact that Boltzmann used, in connection with the Boltzmann equation, the definition 
\be\label{-Hdef}
S=-\kk  \int_{\X_1} \!\! d^3\vq \, d^3\vv \, \tilde{f}(\vq,\vv) \log \tilde{f}(\vq,\vv)
\ee
for entropy. Boltzmann used the notation $H$ (not to be confused with the Hamiltonian function) for the integral in \eqref{-Hdef}; this functional gave name to the $H$-theorem, which asserts that $H$ always decreases---except in thermal equilibrium, when $H$ is constant in time.

Here, $\X_1$ is the 1-particle phase space (for definiteness, $\RRR^6$), $\tilde{f}=Nf$, and $f$ is a normalized distribution density in $\X_1$. The formula \eqref{-Hdef} obviously looks much like the Gibbs entropy (and has presumably inspired Gibbs's definition, which was published after \eqref{-Hdef}). In fact, \eqref{-Hdef} is also the Gibbs entropy of $\rho(x_1,\ldots,x_N) := f(x_1) \cdots f(x_N)$ (up to addition of the constant $\kk N\log N$). But here it is relevant that $f$ means the empirical distribution of $N$ points in $\X_1$ (after smoothing, or in the limit $N\to\infty$), and is thus computed from $X\in \X$ (see Section~\ref{sec:BoltzmannEq} below for more detail). So, the $H$ functional or the $S$ in \eqref{-Hdef} is indeed a function on $\X$, and in fact a special case of \eqref{SBdef} (up to addition of the constant $\kk N$) corresponding to a particular choice of dividing $\X$ into macro sets $\Gamma_\nu$ (see Section~\ref{sec:macrostates}).

\section{The Quantum Case}
\label{sec:quantumcase}

Consider a macroscopic quantum system (e.g., a gas of $N>10^{23}$ atoms in a box) with Hilbert space $\Hilbert$. We write $\SSS(\Hilbert)$ for the unit sphere in the Hilbert space $\Hilbert$. 

\subsection{Entropy in Quantum Mechanics}
\label{sec:SQM}

The natural analog of the Gibbs entropy is the \emph{quantum Gibbs entropy} or \emph{von Neumann entropy} of a given density matrix $\hat\rho$ on $\Hilbert$ \citep{vN27c},
\be\label{SvN}
S_{\vN}(\hat\rho) = -\kk \tr (\hat\rho \log \hat\rho)\,.
\ee
Von Neumann \citeyearpar{vN29} himself thought that this formula was not the most fundamental one but only applicable in certain circumstances; we will discuss his proposal for the fundamental definition in Section~\ref{sec:qm} below. 
In a sense, the density matrix $\hat\rho$ plays the role analogous to the classical distribution density $\rho$, and again, the question arises as to what exactly $\hat\rho$ refers to: an observer's ignorance or what? Our discussion of options (a)--(c) above for the Gibbs entropy will apply equally to the von Neumann entropy. In addition, there is also a further possibility:
\begin{itemize}
\item[(d)] reduced density matrix: The system $\sys_1$ is entangled with another system $\sys_2$, and
$\hat\rho$ is obtained from the (perhaps even pure) state $\Psi\in\SSS(\Hilbert_1\otimes \Hilbert_2)$ of $\sys_1\cup \sys_2$ through a partial trace, $\hat\rho= \tr_2 \pr{\Psi}$. 
\end{itemize}
As we will discuss in Section~\ref{sec:entangle}, $S_\vN$ with option (d) does not yield a good notion of thermodynamic entropy, either.

In practice, systems are never isolated. But \emph{even if} a macroscopic system were perfectly isolated, heat would flow in it from the hotter to the cooler parts, and, as we will explain in this section, there is a natural sense in which entropy can be defined and increases. The idealization of an isolated system helps us focus on this sense. In Section~\ref{sec:open}, we will point out why the results also cover non-isolated systems.

\bigskip

The closest quantum analog of the Boltzmann entropy is the following. A macro state $\nu$ should correspond to, instead of a subset $\Gamma_\nu$ of phase space, a subspace $\Hilbert_\nu$ of Hilbert space $\Hilbert$, called a \emph{macro space} in the following; for different macro states $\nu'\neq \nu$, the macro spaces should be mutually orthogonal, thus yielding a decomposition of Hilbert space into an orthogonal sum \citep{vN29,GLMTZ10,GLTZ10},
\be\label{decomp}
\Hilbert = \bigoplus_\nu \Hilbert_\nu\,,
\ee
instead of the partition \eqref{partition}. Now the dimension of a subspace of $\Hilbert$ plays a role analogous to the volume of a subset of $\X$, and correspondingly we define the \emph{quantum Boltzmann entropy} of a macro state $\nu$ by \citep{Gri,L07,GLTZ10}
\be\label{SqBdef}
S_\qB(\nu) = \kk \log \dim \Hilbert_\nu\,.
\ee
In fact, already \citet[Eq.~(4a)]{Ein14} argued that the entropy of a macro state should be proportional to the log of the ``number of elementary quantum states'' compatible with that macro state; it seems that \eqref{SqBdef} fits very well with this description. To a system with wave function $\psi \in \Hilbert_\nu$ (or, for that matter, a density matrix concentrated in $\Hilbert_\nu$) we attribute the entropy value $S_\qB(\psi) := S_\qB(\nu)$.

It seems convincing that $S_\qB(\nu)$ yields the correct value of thermodynamic entropy. For one thing, it is an extensive, or additive quantity: If we consider two systems $\sys_1,\sys_2$ with negligible interaction, then the Hilbert space of both together is the tensor product, $\Hilbert=\Hilbert_{\sys_1\cup \sys_2}= \Hilbert_{\sys_1} \otimes \Hilbert_{\sys_2}$, and it seems plausible that the macro states of $\sys_1\cup \sys_2$ correspond to specifying the macro state of $\sys_1$ and $\sys_2$, resp., i.e., $\nu=(\nu_1,\nu_2)$ with $\Hilbert_\nu = \Hilbert_{\nu_1} \otimes \Hilbert_{\nu_2}$. As a consequence, the dimensions multiply, so
\be
S_\qB(\nu) = S_\qB(\nu_1) + S_\qB(\nu_2)\,.
\ee
For another thing, it is plausible that, analogously to the classical case, in every energy shell (i.e., the subspace $\Hilbert_{(E-\Delta E,E]}$ corresponding to an energy interval $(E-\Delta E,E]$ with $\Delta E$ the resolution of macroscopic energy measurements) there is a dominant macro space $\Hilbert_{\tilde\nu}$ whose dimension is near (say, greater than $99.99\%$ of) the dimension of the energy shell, and this macro state $\tilde\nu$ corresponds to thermal equilibrium \citep[e.g.,][]{Tas15b}, $\Hilbert_{\tilde \nu}=\Hilbert_\eq$. As a consequence,
\be
S_\qB(\eq) \approx \kk \log \dim \Hilbert_{(E-\Delta E,E]}\,,
\ee
and the right-hand side is well known to yield appropriate values of thermodynamic entropy in thermal equilibrium. In fact, the right-hand side agrees with $S_\vN(\hat\rho_\mc)$, the von Neumann entropy associated with the micro-canonical density matrix $\hat\rho_\mc$ (i.e., the normalized projection to $\Hilbert_{(E-\Delta E,E]}$) and with thermal equilibrium at energy $E$.

\bigskip

Of course, a general pure quantum state $\psi \in \SSS(\Hilbert)$ will be a non-trivial superposition of contributions from different $\Hilbert_\nu$'s, $\psi = \sum_\nu \hat P_\nu \psi$, where $\hat P_\nu$ is the projection to $\Hilbert_\nu$. One can say that $\psi$ is a superposition of different entropy values, and in analogy to other observables one can define the self-adjoint operator
\be\label{hatS}
\hat S = \sum_\nu S_\qB(\nu) \, \hat P_\nu \,,
\ee
whose eigenvalues are the $S_\qB(\nu)$ and eigenspaces the $\Hilbert_\nu$. 

Here, the question arises as to which entropy value we should attribute to a system in state $\psi$. At this point, the foundations of statistical mechanics touch the foundations of quantum mechanics, as the problem of Schr\"odinger's cat (or the measurement problem of quantum mechanics) concerns precisely the status of wave functions $\psi$ that are superpositions of macroscopically different contributions, given that our intuition leads us to expect a definite macroscopic situation. The standard ``Copenhagen'' formulation of quantum mechanics does not have much useful to say here, but several proposed ``quantum theories without observers'' have long solved the issue in clean and clear ways \citep{Gol98}: 
\begin{itemize}
\item Bohmian mechanics (in our view the most convincing proposal) admits further variables besides $\psi$ by assuming that quantum particles have definite positions, too. Since Schr\"odinger's cat then has an actual configuration, there is a fact about whether it is dead or alive, even though the wave function is still a superposition. In the same way, the configuration usually selects one of the macroscopically different contributions $\hat P_\nu \psi$; in fact, this happens when
\be\label{dontoverlap}
\text{the significantly nonzero $\hat P_\nu \psi$ do not overlap much in configuration space.}
\ee
Artificial examples can be constructed for which this condition is violated, but it seems that \eqref{dontoverlap} is fulfilled in practice. In that case, it seems natural to regard the entropy of that one contribution selected by the Bohmian configuration as the actual entropy value.
\item Collapse theories modify the Schr\"odinger equation in a non-linear and stochastic way so that, for macroscopic systems, the evolution of $\psi$ avoids macroscopic superpositions and drives $\psi$ towards (a random) one of the $\Hilbert_\nu$. Then the question about the entropy of a macroscopic superposition does not arise.
\item Maybe a many-worlds view \citep{Eve57,Wal12}, in which each of the contributions corresponds to a part of reality, can be viable (see \citet{AGTZ11} for critical discussion). Then the system's entropy has different values in different worlds.
\end{itemize}

We have thus defined the quantum Boltzmann entropy of a macroscopic system,
\be\label{SqBsys}
S_\qB(\sys)\,,
\ee
in each of these theories.

Some authors (e.g., \citet{vN29,SDA17,SDA18}) have proposed averaging $S_\qB(\nu)$ with weights $\|\hat P_\nu \psi\|^2$, but that would yield an \emph{average} entropy, not \emph{the} entropy (see also Section~\ref{sec:wrongvalues} and our discussion of \eqref{SvN2} in Section~\ref{sec:historical}).

\bigskip

\label{sec:2ndquantum}

Let us turn to the question of what the second law asserts in quantum mechanics. Since $S_\vN$ remains constant under the unitary time evolution in $\Hilbert$, the discussion given in Section~\ref{sec:indep} for $S_\G$ applies equally to $S_\vN$. About $S_\qB$, we expect the following.

\begin{conj}\label{conj:qB2ndlaw}
In Bohmian mechanics and collapse theories, for macroscopic systems $\sys$ with reasonable Hamiltonian $\hat H$ and decomposition $\Hilbert= \oplus_\nu\Hilbert_\nu$, every $\nu^*$ and most $\psi\in\SSS(\Hilbert_{\nu^*})$ are such that with probability close to 1, $S_\qB(\sys)$ is non-decreasing with time, except for infrequent, shallow, and short-lived valleys, until it reaches the thermal equilibrium value. Moreover, after reaching that value, $S_\qB(\sys)$ stays at that value for a very long time.
\end{conj}

Careful studies of this conjecture would be of interest but are presently missing. We provide a bit more discussion in Section~\ref{sec:qm}.

\subsection{Open Systems}
\label{sec:open}

For open (non-isolated) systems, the quantum Boltzmann entropy can still be considered and should still tend to increase---not for the reasons that make the von Neumann entropy of the reduced density matrix increase (evolution from pure to mixed states), but for the reasons that make the quantum Boltzmann entropy of an isolated system in a pure state increase (evolution from small to large macro spaces). Let us elaborate. 

In fact, the first question about an open system $\sys_1$ (interacting with its environment $\sys_2$) would be how to define its state. A common answer is the \emph{reduced density matrix}
\be\label{reddef}
\hat \rho_1:= \tr_2 |\Psi\rangle \langle \Psi|
\ee
(if $\sys_1\cup\sys_2$ is in the pure state $\Psi$); another possible answer is the \emph{conditional density matrix} \citep{DGTZ05}
\be\label{conddef}
\hat \rho_\mathrm{cond}:= \tr_{s_2} \langle x_2|\Psi\rangle \langle \Psi|x_2\rangle\Big|_{x_2=X_2}\,,
\ee
where $\tr_{s_2}$ denotes the partial trace over the non-positional degrees of freedom of $\sys_2$ (such as spin), the scalar products are partial scalar products involving only the position degrees of freedom of $\sys_2$, and $X_2$ is the Bohmian configuration of $\sys_2$. (To illustrate the difference between $\hat \rho_1$ and $\hat\rho_\mathrm{cond}$, if $\sys_1$ is Schr\"odinger's cat and $\sys_2$ its environment, then $\hat\rho_1$ is a mixture of a live and a dead state, whereas $\hat\rho_\mathrm{cond}$ is either a live or a dead state, each with the appropriate probabilities.)

Second, it seems that the following analog of \eqref{dontoverlap} is usually fulfilled in practice for both the reduced density matrix $\hat\rho=\hat\rho_1$ and the conditional density matrix $\hat\rho=\hat\rho_\mathrm{cond}$:
\be\label{dontoverlaprho}
\mbox{\begin{minipage}{0.85\textwidth}
For those $\nu$ for which $\hat P_\nu \hat\rho \hat P_\nu$ is significantly nonzero,
the functions $g_\nu$ on configuration space given by $g_\nu(x):= \langle x|\hat P_\nu \hat\rho \hat P_\nu | x\rangle$ do not overlap much.
\end{minipage}}
\ee
As with \eqref{dontoverlap}, if \eqref{dontoverlaprho} holds, then the Bohmian configuration $X_1$ of $\sys_1$ selects the actual macro state $\nu(X_1)$ of $\sys_1$. Moreover, $\nu(X_1)$ should usually be the same for $\hat\rho=\hat\rho_1$ as for $\hat\rho=\hat\rho_\mathrm{cond}$. Then this macro state determines the entropy,
\be
S_\B(\sys_1) := S_\B(\nu(X_1))\,.
\ee
In short, the concept of quantum Boltzmann entropy carries over from isolated systems in pure states to open systems.

\subsection{Quantum Thermalization}
\label{sec:thermalization}

In the 21st century, there has been a wave of works on the thermalization of closed quantum systems, often connected with the key words ``eigenstate thermalization hypothesis'' (ETH) and ``canonical typicality''; see, e.g., \citep{GMM04,GLTZ06,PSW06,GLMTZ10,GE15,GHLT15,KTL16} and the references therein. The common theme of these works is that an individual, closed, macroscopic quantum system in a pure state $\psi_t$ that evolves unitarily will, under conditions usually satisfied and after a sufficient waiting time, behave very much as one would expect a system in thermal equilibrium to behave; more precisely, on such a system $\sys$ with $\psi_0$ in an energy shell, relevant observables yield their thermal equilibrium values up to small deviations with probabilities close to 1. For example, this happens simultaneously for all observables referring to a small subsystem $\sys_1$ of $\sys$ that interacts weakly with the remainder $\sys_2:= \sys \setminus \sys_1$, with the further consequence (``canonical typicality'') that the reduced density matrix of $\sys_1$,
\be
\hat\rho_1 := \tr_{2} \pr{\psi_t}\,,
\ee
is close to a canonical one \citep{GLTZ06},
\be
\hat\rho_1 \approx \frac{1}{Z} e^{-\beta \hat H_1}\,,
\ee
for suitable $\beta$ and normalizing constant $Z$; here, $\hat H_1$ is the Hamiltonian of $\sys_1$. For another example, every initial wave function of $\sys$ in an energy shell will, after a sufficient waiting time and for most of the time in the long run, be close to the thermal equilibrium macro space $\Hilbert_\eq$ \citep{GLMTZ10}, provided that the Hamiltonian $\hat H$ is non-degenerate and (ETH) all eigenvectors of $\hat H$ are very close to $\Hilbert_{\eq}$.

These works support the idea that the approach to thermal equilibrium need not have anything to do with an observer's ignorance. In fact, the system $\sys$ always remains in a pure state, and thus has von Neumann entropy $S_\vN=0$ at all times. This fact illustrates that the kind of thermalization relevant here involve neither an increase in von Neumann entropy nor a stationary density matrix of $\sys$. Rather, $\psi_t$ reaches, after a sufficient waiting time, the $\varepsilon$-neighborhood of the macro space $\Hilbert_{\tilde\nu}=\Hilbert_\eq$ corresponding to thermal equilibrium in the energy shell \citep{GLMTZ10}. That is the ``individualist'' or ``Boltzmannian'' version of approach to thermal equilibrium in quantum mechanics.

In fact, there are two individualist notions of thermal equilibrium in quantum mechanics, which have been called ``macroscopic'' and ``microscopic thermal equilibrium'' \citep{GHLT17}. Boltzmann's approach requires only that macro observables assume their equilibrium values \citep{GLMTZ10}, whereas a stronger statement is actually true after a long waiting time: that all micro observables assume their equilibrium values \citep{GHLT15}. This is true not only for macroscopic systems, but also for small systems \citep{GHLS17}, and has in fact been verified experimentally for a system with as few as 6 degrees of freedom \citep{KTL16}.

We have emphasized earlier in this subsection that thermalization does not require mixed quantum states. We should add that this does not mean that pure quantum states are fully analogous to phase points in classical mechanics. In Bohmian mechanics, for example, the analog of a phase point would be the pair $(\psi,Q)$ comprising the system's wave function $\psi$ and its configuration $Q$.

\section{Subjective Entropy Is Not Enough}
\label{sec:subjective}

By ``subjective entropy'' we mean Gibbs entropy under option (a): the view that the Gibbs entropy is the thermodynamic entropy, and that the distribution $\rho$ in the Gibbs entropy represents an observer's subjective perspective and limited knowledge \citep[e.g.,][]{Jay65,Kry,Mac89,Teg}. We would like to point to three major problems with this view.

\subsection{Cases of Wrong Values}
\label{sec:wrongvalues}

The first problem is that in some situations, the subjective entropy does not appropriately reproduce the thermodynamic entropy. For example, suppose an isolated room contains a battery-powered heater, and we do not know whether it is on or off. If it is on, then after ten minutes the air will be hot, the battery empty, and the entropy of the room has a high value $S_3$. Not so if the heater is off; then the entropy has the low initial value $S_1<S_3$. In view of our ignorance, we may attribute a subjective probability of 50 percent to each of ``on'' and ``off.'' After ten minutes, our subjective distribution $\rho$ over phase space will be spread over two regions with macroscopically different phase points, and its Gibbs entropy $S_\G(\rho)$ will have a value $S_2$ between $S_1$ and $S_3$ (in fact, slightly above the average of $S_1$ and $S_3$).\footnote{Generally, if several distribution functions $\rho_i$ have mutually disjoint supports, and we choose one of them randomly with probability $p_i$, then the resulting distribution $\rho=\sum_i p_i \rho_i$ has Gibbs entropy $S_\G(\sum_i p_i \rho_i)
= -\kk\int dx (\sum_i p_i\, \rho_i(x)) \log (\sum_i p_i\, \rho_i(x))
= -\kk\sum_i \int dx \, p_i \, \rho_i \log (p_i\, \rho_i(x)) = \sum_i p_i\, S_\G(\rho_i) -\kk \sum_i p_i \,\log p_i > \sum_i p_i S_\G(\rho_i)$. In our example, $S_2=(S_1+S_3)/2 + \kk\log 2$.} But the correct thermodynamic value is not $S_2$, it is either $S_1$ (if the heater was off) or $S_3$ (if the heater was on). So subjective entropy yields the wrong value.

The same problem arises with option (b), which concerns a system prepared by some procedure that, when repeated over and over, will lead to a distribution $\rho$ of the system's phase point $X$. Suppose that the isolated room also contains a mechanism that tosses a coin or generates a random bit $Y\in\{0,1\}$ in some other mechanical way; after that, the mechanism turns on the heater or does not, depending on $Y$. We would normally say that the entropy $S$ after ten minutes is random, that $S=S_1$ with probability 1/2 and $S=S_3$ with probability 1/2. But the distribution $\rho$ created by the procedure (and the canonical distribution for each given value of $Y$) is the same as in the previous paragraph, and has Gibbs entropy $S_\G(\rho)=S_2$, the wrong value.

\subsection{Explanatory and Predictive Power}

The second major problem with subjective entropy concerns the inadequacy of its explanatory and predictive power. Consider for example the phenomenon that by thermal contact, heat always flows from the hotter to the cooler body, not the other way around. The usual explanation of this phenomenon is that entropy decreases when heat flows to the hotter body, and the second law excludes that. Now that explanation would not get off the ground if entropy meant subjective entropy: In the absence of observers, does heat flow from the cooler to the hotter? In distant stars, does heat flow from the cooler to the hotter? In the days before humans existed, did heat flow from the cooler to the hotter? After the human race becomes extinct, will heat flow from the cooler to the hotter? If not, why would observers be relevant at all to the explanation of the phenomenon? 

And as with explanation, so with prediction: Can we predict that heat will flow from the hotter to the cooler also in the absence of observers? If so, why would observers be relevant to the prediction?

So, subjective entropy does not seem relevant to either explaining or predicting heat flow. That leaves us with the question, what is subjective entropy good for? The study of subjective entropy is a subfield of psychology, not of physics. It is all about beliefs.

Some ensemblists may be inclined to say that the explanation of heat flow is that it occurs the same way \emph{as if} an observer observed it. But the fact remains that observers actually have nothing to do with it.

Once the problem of explanatory power is appreciated, it seems obvious that subjective entropy is inappropriate: How could an objective physical phenomenon such as heat flow from the hotter to the cooler depend on subjective belief? In fact, since different observers may have different, incompatible subjective beliefs, how could coherent consequences such as physical laws be drawn from them? And what if the subjects made mistakes, what if they computed the time-evolved distribution $\rho\circ \Phi_t$ incorrectly, what if their beliefs were irrational---would that end the validity of subjective entropy? Somebody may be inclined to say that subjective entropy is valid only if it is rational \citep[e.g.,][]{Bri18}, but that means basically to back off from the thought that entropy is subjective. It means that it does not play much of a role whether anybody's actual beliefs follow that particular $\rho$, but rather that there is a correct $\rho$ that should be used; we will come back to this view at the end of the next subsection.

Another drawback of the subjective entropy, not unrelated to the problem of explanatory power, is that it draws the attention away from the fact that the universe must have very special initial conditions in order to yield a history with a thermodynamic arrow of time. While the Boltzmann entropy draws attention to the special properties of the initial state of the universe, the subjective entropy hides any such objective properties under talk about knowledge.

\subsection{Phase Points Play No Role}

The third problem with subjective entropy is that $S_\G(\rho)$ has nothing to do with the properties of the phase points $x$ at which $\rho$ is significantly non-zero. $S_\G(\rho)$ measures essentially the width of the distribution $\rho$, much like the standard deviation of a probability distribution, except that the standard deviation yields the radius of the set over which $\rho$ is effectively distributed, whereas the Gibbs entropy yields the log of its volume (see \eqref{SGA} below). The problem is reflected in the fact, mentioned around \eqref{SGPhi}, that any volume-preserving transformation $\Phi:\X\to\X$ will leave the Gibbs entropy unchanged. It does not matter to the Gibbs entropy how the $x$'s on which $\rho$ is concentrated behave physically, although this behavior is crucial to thermodynamic entropy and the second law.

Some ensemblists may be inclined to say that the kind of $\rho$ that occurs in practice is not any old density function, but is approximately concentrated on phase points that look macroscopically similar. This idea is essentially option (c) of Section~\ref{sec:xrho}, which was to take $\rho$ as a kind of coarse graining of the actual phase point $X$. Specifically, if $\Gamma(X)$ denotes again the set of phase points that look macroscopically similar to $X$, then we may want to take $\rho=\rho_X$ to be the flat distribution over $\Gamma(X)$,
\be
\rho_X(x) = \frac{1}{\vol\Gamma(X)} 1_{\Gamma(X)}(x) = \begin{cases} 
(\vol \Gamma(X))^{-1} & \text{if }x\in \Gamma(X)\\ 0 &\text{if }x\notin \Gamma(X).
\end{cases}
\ee
With this choice we obtain exact agreement between the Gibbs and Boltzmann entropies,
\be
S_\G(\rho_X) = S_\B(X)\,.
\ee
Indeed, whenever $\rho$ is the flat distribution over any subset $A$ of phase space $\X$, $\rho(x) = (\vol A)^{-1} \, 1_A(x)$, then
\be\label{SGA}
S_\G(\rho) = - \kk \! \int_A \!\! dx \, \frac{(-1)}{\vol A} \log \vol A = \kk \log \vol A\,.
\ee
(This fact also illustrates the mathematical meaning of the Gibbs entropy of any distribution $\rho$ as the log of the volume over which $\rho$ is effectively distributed.)

Of course, if we associate an entropy value $S_\G(\rho_X)$ with every $X\in\X$ in this way, then the use of Gibbs's definition \eqref{SGdef} seems like an unnecessary detour. In fact, we have associated with every $X\in\X$ an entropy value $S(X)$, and talk about the knowledge of observers is not crucial to the definition of the function $S$, as is obvious from the fact that the function $S$ is nothing but the Boltzmann entropy, which was introduced without mentioning observers. 

This brings us once more to the idea that the $\rho$ in $S_\G(\rho)$ is the subjective belief of a \emph{rational} observer as advocated by \citet{Bri18}. One could always use the Boltzmann entropy and add a narrative about observers and their beliefs, such as: Whenever $X\in \Gamma_\nu$, a rational observer should use the flat distribution over $\Gamma_\nu$, and the Gibbs entropy of that observer's belief is what entropy really means. One could say such words. But they are also irrelevant, as observers' knowledge is irrelevant to which way heat flows, and the resulting entropy value agrees with $S_\B(X)$ anyway.

\subsection{What is Attractive About Subjective Entropy}

Let us turn to factors that may seem attractive about the subjective entropy view: First, it seems like the obvious interpretation of the density $\rho$ that comes up in all ensembles. But the Boltzmannian individualist view offers an alternative interpretation, as we explain in Section~\ref{sec:vision}. 

Second, it is simple and elegant. That may be true but does not do away with the problems mentioned. 

Third, the subjective view mixes well with certain interpretations of quantum mechanics such as Copenhagen and quantum Bayesianism, which claim that quantum mechanics is all about information or that one should not talk about reality. These interpretations are problematical as well, mainly because all information must ultimately be information about things that exist, and it does not help to leave vague and unspecified which things actually exist \citep{Gol98}. 

Fourth, the subjective view may seem to mix well with the work of \citet{Sha}, as the Shannon entropy is a discrete version of Gibbs entropy and often regarded as quantifying the information content of a probability distribution. But actually, there is not a strong link, as Shannon regarded the probabilities in his study of optimal coding of data for transmission across a noisy channel as objective and did not claim any connection with thermodynamics. (By the way, it is dangerous to loosely speak of the ``amount of information'' in the same way as one speaks of, e.g., the amount of sand; after all, the sand grains are equal to each other, and one does not care about whether one gets this or that grain, whereas different pieces of information are not equivalent to each other.)

Fifth and finally, a strong pull towards subjective entropy comes from the belief that ``objective entropy'' either does not work or is ugly---a wrong belief, as we will explain in Section~\ref{sec:vision}.

\subsection{Remarks}

Further critiques of subjective entropy can be found in \citep{Cal99,LM03,GHLS17,Gol19}. 

We would like to comment on another quote. \citet{Jay65}, a defender of subjective entropy, reported a dictum of Wigner's:
\begin{quotation}
Entropy is an anthropomorphic concept.
\end{quotation}
Of course, this phrase can be interpreted in very different ways. Jaynes took it to express that entropy refers to the knowledge of human observers---the subjective view that we have criticized. But we do admit that there is a trait in entropy that depends partly on human nature, and that is linked to a certain (though limited and usually unproblematical) degree of arbitrariness in the definition of ``looking macroscopically the same.'' This point will come up again in the next section.

\section{Boltzmann's Vision Works}
\label{sec:vision}

Many authors expressed disbelief that Boltzmann's understanding of entropy and the second law could possibly work. Von Neumann \citeyearpar[Sec.~0.6]{vN29} wrote:
\begin{quotation}
As in classical mechanics, also here [in the quantum case] there is no way that entropy could always increase, or even have a predominantly positive sign of its [time] derivative (or difference quotient): the time reversal objection as well as the recurrence objection are valid in quantum mechanics as well as in classical mechanics.
\end{quotation}
\citet[\S 33, p.~139]{Khi}:
\begin{quotation}
[One often] states that because of thermal interaction of material bodies the entropy of the universe is constantly increasing. It is also stated that the entropy of a system ``which is left to itself'' must always increase; taking into account the probabilistic foundation of thermodynamics, one often ascribes to this statement a statistical rather than absolute character. This formulation is wrong if only becuase the entropy of an isolated system is a thermodynamic function---not a phase-function---which means that it cannot be considered as a random quantity; if $E$ and all [external parameters] $\lambda_s$ remain constant the entropy cannot change its value whereas by changing these parameteres in an appropriate way we can make the entropy increase or decrease at will. Some authors (footnote: Comp.\ Borel, M\'ecanique statistique classique, Paris 1925.)\ try to generalize the notion of entropy by considering it as being a phase function which, depending on the phase, can assume different values for the same set of thermodynamical parameters, and try to prove that entropy so defined must increase, with overwhelming probability. However, such a proof has not yet been given, and it is not at all clear how such an artificial generalization of the notion of entropy could be useful to the science of thermodynamics.
\end{quotation}
\citet{Jay65}:
\begin{quotation}
[T]he Boltzmann $H$ theorem does not constitute a demonstration of the second law for dilute gases[.]
\end{quotation}
Even Boltzmann himself was at times unassured.
In a letter to Felix Klein in 1889, he wrote:
\begin{quotation}
Just when I received your dear letter I had another neurasthenic attack, as I often do in Vienna, although I was spared them altogether in Munich. With it came the fear that the whole $H$-curve was nonsense.
\end{quotation}

But actually, the $H$-curve (i.e., the time evolution of entropy) makes complete sense, Boltzmann's vision does work, and von Neumann, Khinchin, and Jaynes were all mistaken, so it is perhaps worth elucidating this point. Many other, deeper discussions can be found in the literature, e.g., qualitative, popular accounts in \citep{Pen,LM03,Car}, overviews in \citep{Gol99,L07,Gol19}, more technical and detailed discussions in \citep{Bol1898,ehrenfest,Lan76,GGL04,FPPRV07,CCCV16,GHLS17,Laz}. So we now give a summary of Boltzmann's explanation of the second law.

\subsection{Macro States}
\label{sec:macrostates}

We start with a partition $\X= \cup_\nu \Gamma_\nu$ of phase space into macro sets as in Figure~\ref{fig:phasespace}. A natural way of obtaining such a partition would be to consider several functions $M_j:\X\to\RRR$ ($j=1,\ldots,K$) that we would regard as ``macro variables.'' Since macro measurements have limited resolution (say, $\Delta M_j>0$), we want to think of the $M_j$ as suitably coarse-grained with a discrete set of values, say, $\{n\Delta M_j: n\in \ZZZ\}$. Then two phase points $x_1,x_2\in\X$ will look macroscopically the same if and only if $M_j(x_1)=M_j(x_2)$ for all $j=1,\ldots, K$, corresponding to
\be\label{GammaMdef}
\Gamma_\nu = \Bigl\{x\in \X: M_j(x)=\nu_j \:\forall j\Bigr\}\,,
\ee
one for every macro state $\nu=(\nu_1,\ldots,\nu_K)$ described by the list of values of all $M_j$. We will discuss a concrete example due to Boltzmann in Section~\ref{sec:BoltzmannEq}. Since coarse-grained energy should be one of the macro variables, say
\be\label{M1}
M_1(x) = [H(x)/\Delta E] \Delta E
\ee
with $H$ the Hamiltonian function and $[x]$ the nearest integer to $x\in\RRR$, every $\Gamma_\nu$ is contained in one micro-canonical energy shell
\be\label{Xmc}
\X_\mc := \X_{(E-\Delta E,E]} := \Bigl\{ x\in\X: E-\Delta E < H(x) \leq E \Bigr\}\,.
\ee
Of course, this description still leaves quite some freedom of choice and thus arbitrariness in the partition, as different physicists may make different choices of macro variables and of the way and scale to coarse-grain them; this realization makes an ``anthropomorphic'' element in $S_\B$ explicit. \citet{Wal18} complained that this element makes the Boltzmann entropy ``subjective'' as well, but that complaint does not seem valid: rather, $S_\B$ and its increase provide an objective answer to a question that is of interest from the human perspective. Moreover, as mentioned already, this anthorpomorphic element becomes less relevant for larger $N$. It is usually not problematical and not subject to the same problems as the subjective entropy.

Usually in macroscopic systems, there is, for every energy shell $\X_\mc$ (or, if there are further macroscopic conserved quantities besides energy, in the set where their values have been fixed as well), one macro set $\Gamma_{\tilde\nu}$ that contains most (say, more than $99.99\%$) of the phase space volume of $\X_\mc$ \citep[see, e.g.,][]{Bol1898,Lan73,Laz};\footnote{\label{fn:ferro}There are exceptions, in which none of the macro sets dominates; for example, in the ferromagnetic Ising model with vanishing external magnetic field and not-too-high temperature, there are two macro states (the first having a majority of spins up, the second having a majority of spins down) that together dominate but have equal volume; see also \citep{Laz}. But that does not change much about the discussion.} in fact \citep[Eq.~(6)]{GHLT17},
\be\label{10cN}
\frac{\vol \Gamma_{\tilde\nu}}{\vol \X_\mc} \approx 1-  10^{-cN}
\ee
with positive constant $c$. The existence of this dominant macro state means that all macro observables are nearly-constant functions on $\X_{\mc}$, in the sense that the set where they deviate (beyond tolerances) from their dominant values has tiny volume. This macro state $\tilde\nu$ is the thermal equilibrium state, $\Gamma_{\tilde\nu}=\Gamma_{\eq}$, see Figure~\ref{fig:phasespace}, and the dominant values of the macro observables are their thermal equilibrium values. That fits with thermal equilibrium having maximal entropy, and it has the consequence that
\be\label{SBeqSGmc}
S_\B(\eq) \approx \kk \log \vol \X_{\mc} = S_\G(\rho_\mc)\,,
\ee
where $\rho_\mc(x)= (\vol \X_\mc)^{-1} \, 1_{\X_\mc}(x)$ is the micro-canonical distribution, and the (relative or absolute) error in the approximation tends to zero as $N\to\infty$.

Moreover, different macro sets $\Gamma_\nu$ have vastly different volumes. In fact, usually the small macro sets of an energy shell taken together are still much smaller than the next bigger one,
\be\label{<nu<<nu}
\vol(\Gamma_{<\nu}^{(E-\Delta E,E]})\ll \vol(\Gamma_{\nu})
\ee
with
\be
\Gamma_{<\nu}^{(E-\Delta E,E]} := \bigcup_{\nu': S_\B(\nu')<S_\B(\nu)} \Gamma_{\nu'}\cap \X_{(E-\Delta E,E]} \,.
\ee
(There are exceptions to this rule of thumb; in particular, symmetries sometimes imply that two or a few macro sets must have approximately the same volume.)

\subsection{Entropy Increase}

Now increase of Boltzmann entropy means that the phase point $X(t)$ moves to bigger and bigger macro sets $\Gamma_\nu$. In this framework, the second law can be stated as follows.
\be\label{2ndlaw}
\mbox{\begin{minipage}{0.85\textwidth}
{\bf Mathematical second law.} Given $\nu \neq \eq$, for most phase points $X(0)$ in $\Gamma_\nu$, $X(t)$ moves to bigger and bigger macro sets as $t$ increases until it reaches $\Gamma_\eq$, except possibly for entropy valleys that are infrequent, shallow, and short-lived; once $X(t)$ reaches $\Gamma_\eq$, it stays in there for an extraordinarily long time, except possibly for infrequent, shallow, and short-lived entropy valleys.
\end{minipage}}
\ee
The described behavior is depicted in Figure~\ref{fig:S(t)}. Entropy valleys (i.e., periods of entropy decrease and return to the previous level) are also called fluctuations. The \emph{physical second law} then asserts that the actual phase point of a real-world closed system behaves the way described in \eqref{2ndlaw} for most phase points.

\begin{figure}[h]
\begin{center}
\includegraphics[width=.9 \textwidth]{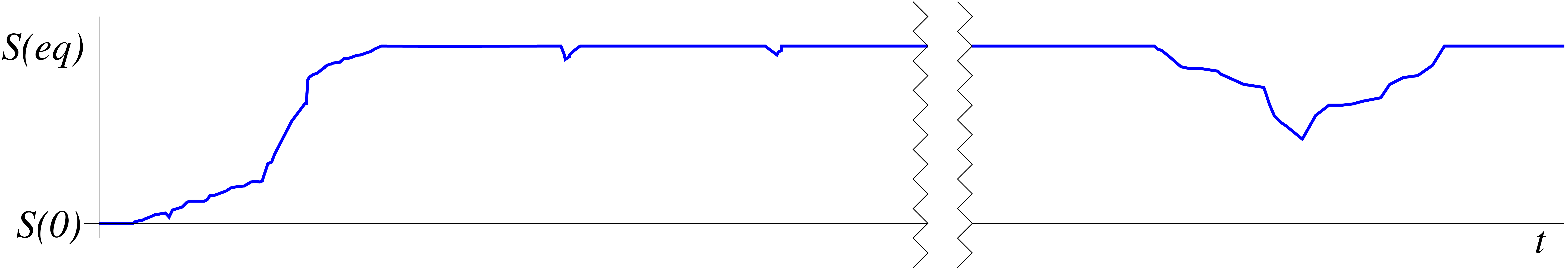}
\end{center}
\caption{A typical entropy curve $S(X(t))$ according to Boltzmann: It should go up except for infrequent, shallow, short-lived valleys; frequency, depth, and duration of the valleys are exaggerated for better visibility. After very long times, the entropy should go down considerably.}
\label{fig:S(t)}
\end{figure}

As an illustration of \eqref{2ndlaw} and as a step towards making it plausible, let us consider two times, 0 and $t>0$. Let $A_t := \Phi_t(\Gamma_\nu)$. By Liouville's theorem, $\vol(A_t)=\vol(\Gamma_\nu)$, and thus, by \eqref{<nu<<nu},
\be
\frac{\vol(A_t \cap \Gamma^{(E-\Delta E,E]}_{<\nu})}{\vol(A_t)}\ll 1\,.
\ee
That is, only a small minority of points in $A_t$ will have entropy smaller than $S_\B(\nu)$. That is, for most points $X(0)\in\Gamma_\nu$,
\be
S_\B(X(0)) \leq S_\B(X(t)) \,.
\ee

Another simple special case is the one in which the macro evolution is deterministic \citep{GGL04,GL04,DRMN06}. For the sake of concreteness, assume that in a time step of a certain size $\tau$, $\Gamma_{\nu_1}$ gets mapped into $\Gamma_{\nu_2}$, which in turn gets mapped into $\Gamma_{\nu_3}$, and so on up to $\nu_m$:
\be
\Phi_\tau(\Gamma_{\nu_i}) \subseteq \Gamma_{\nu_{i+1}}\,.
\ee
Then, by Liouville's theorem, $\vol \Gamma_{\nu_i}  \leq \vol \Gamma_{\nu_{i+1}}$, so
\be
S_\B(\nu_i) \leq S_\B(\nu_{i+1})
\ee
for all $i=1,\ldots, m-1$, so entropy does not decrease. Of course, in realistic cases, the macro evolution becomes deterministic only in the limit $N\to\infty$, and as long as $N$ is finite, there are a minority of points in $\Gamma_{\nu_i}$ that do not evolve to $\Gamma_{\nu_{i+1}}$.

Generally, if the Hamiltonian motion is not specially desgined for the given partition of $\X$, then it is quite intuitive that the motion of the phase point should tend to lead to larger macro sets, and not to smaller ones. Numerical simulations exhibiting this behavior are presented in \citep{FPPRV07}. It is also quite intuitive that the phase point would stay in $\Gamma_\eq$ for a very, very long time: If the non-equilibrium set $\Gamma_\mathrm{noneq}:=\Gamma^{(E-\Delta E,E]}_{<\eq}= \X_\mc \setminus \Gamma_\eq$ has only the fraction $10^{-cN}$ of the volume of the energy shell, cf.\ \eqref{10cN}, then only a tiny fraction of $\Gamma_\eq$ should be able to evolve into the non-equilibrium set $\Gamma_\mathrm{noneq}$ in a short time; and if most points in $\Gamma_\mathrm{noneq}$ spend a substantial amount of time there, then it will take very, very long until a substantial fraction of $\Gamma_\eq$ has visited $\Gamma_\mathrm{noneq}$. The statement that points in $\Gamma_\eq$ stay there for a long time fits well with the observed stationarity of thermal equilibrium---which is why it is called ``equilibrium.''

Let us briefly address two classic objections to the idea that entropy increases:
\begin{itemize}
\item Time reversal (Loschmidt's objection) shows that entropy increase cannot hold for \emph{all} phase points in $\Gamma_\nu$. Concretely, for relevant Hamiltonians the time reversal mapping $R:\X\to\X$, defined by
\be
R(\vq_1,\ldots,\vq_N,\vv_1,\ldots,\vv_N):= (\vq_1,\ldots,\vq_N, -\vv_1,\ldots,-\vv_N)
\ee
with $\vq_i$ the position and $\vv_i$ the velocity of particle $i$, has the property
\be
R\circ \Phi_t \circ R = \Phi_{-t}\,.
\ee
Usually, $R$ maps $\Gamma_\nu$ onto some $\Gamma_{\nu'}$ (where $\nu'$ may or may not equal $\nu$), so $S_\B(R(x))=S_\B(x)$. So if some $X(0)\in \Gamma_{\nu_1}$ evolves to $X(t) \in \Gamma_{\nu'}$ with $S_\B(\nu') > S_\B(\nu_1)$, then $R(X(t))\in \Gamma_\nu$ evolves to $R(X(0))$, and its entropy decreases.

\item Recurrence (Zermelo's objection) shows that $S_\B(X(t))$ cannot \emph{forever} be non-decreasing; thus, $X(t)$ cannot stay forever in $\Gamma_\eq$ once it reaches $\Gamma_\eq$. (The Poincar\'e recurrence theorem states that under conditions usually satisfied in $\X_\mc$, every trajectory $X(t)$, except for a set of measure zero of $X(0)$s, returns arbitrarily close to $X(0)$ at some arbitrarily late time.)
\end{itemize} 
Contrary to von Neumann's statement quoted in the beginning of Section~\ref{sec:vision}, the second law as formulated in \eqref{2ndlaw} is not refuted by either objection: after all, the second law applies to \emph{most}, not \emph{all}, phase points $X(0)$, and it does not claim that $X(t)$ will stay in thermal equilibrium \emph{forever}, but only for a very, very long time.

\subsection{Non-Equilibrium}

The term ``non-equilibrium'' is sometimes understood \citep{Gal} as referring to so called non-equilibrium steady states (NESS), which concerns, for example, a system $\sys$ coupled to two infinite reservoirs of different temperature; so $\sys$ is an open system heated on one side and cooled on another, and it will tend to assume a macroscopically stationary (``steady'') state with a temperature gradient, a nonzero heat current, and a positive rate of entropy production \citep{Ons31,BL55,Der07,GHLS17}. In contrast, in this Section~\ref{sec:vision} we are considering a closed system (i.e., not interacting with the outside), and ``non-equilibrium'' refers to any phase point in $\X_\mc \setminus \Gamma_\eq$. Examples of non-equilibrium macro states include, but are not limited to, states in local thermal equilibrium but not in (global) thermal equilibrium (such as systems hotter in one place than in another). Other examples arise from removing a constraint or wall; such macro states may have been in thermal equilibrium before the constraint was removed but are not longer so afterwards; for example, think of a macro state in which all particles are in the left half of a box; for another example, suppose we could turn on or off the interaction between two kinds of particles (say, ``red ones'' and ``blue ones''), and think of a macro state that is a thermal equilibrium state when the interaction is off (so that the red energy and the blue energy are separately conserved)  but not when it is on, such as when both gases are in the same volume but at different temperatures.

\subsection{Concrete Example: The Boltzmann Equation}
\label{sec:BoltzmannEq}

Here is a concrete example of a partition of phase space due to Boltzmann. Divide the 1-particle phase space $\X_1$ into cells $C_i$ (say, small cubes of equal volume) and count (with a given tolerance) the particles in each cell. The macro state is described by the list $\nu=(\nu_1,\ldots,\nu_L)$ of all these occupation numbers; for convenience, we will normalize them:
\be\label{occupation}
\nu_i(x) := f_i:=  \Biggl[\frac{\#\bigl\{j\in \{1,\ldots,N\}: (\vq_j,\vv_j)\in C_i \bigr\} }{N \, \Delta f \, \vol C_i} \Biggr]\Delta f 
\ee
with $N\Delta f$ the tolerance in counting and $[\cdot]$ again the nearest integer. This example of a partition $\Gamma_\nu=\{x\in\X: \nu(x)=\nu\}$ is good for dilute, weakly interacting gases but not in general \citep{GGL04,GL04} (see also Section~\ref{sec:objections}). 

Boltzmann considered $N$ billiard balls of radius $a$ in a container $\Lambda\subset\RRR^3$, so $\X_1=\Lambda \times \RRR^3$. In a suitable limit in which $N\to\infty$, $a\to 0$, and the cells $C_i$ become small, the normalized occupation numbers $f_i$ become a continuous density $f(\vq,\vv)$. He argued (convincingly) that for most $x\in \Gamma_\nu$ this density, essentially the empirical distribution of the $N$ particles in $\X_1$, will change in time according to the \emph{Boltzmann equation}, an integro-differential equation \citep{Bol1872,Bol1898,ehrenfest,Lan76}. It reads, in the version appropriate for the hard sphere gas without external forces,
\be
\label{BoltzmannEq}
\Bigl( \frac{\partial}{\partial t} + \vv \cdot \nabla_{\vq} \Bigr) f(\vq,\vv,t) 
= Q(\vq,\vv,t)
\ee
with the ``collision term''
\begin{multline}\label{collisionterm}
Q(\vq,\vv,t) =\lambda \int_{\RRR^3} d^3\vv_*\int_{\SSS^2} d^2\vomega\:1_{\vomega\cdot(\vv-\vv_*)>0} \: \vomega\cdot (\vv-\vv_*) \:\times\\
\Bigl[ f(\vq,\vv',t) \, f(\vq,\vv'_*,t) - f(\vq, \vv,t) \, f(\vq,\vv_*,t) \Bigr]\,,
\end{multline}
involving a constant $\lambda>0$ and the abbreviations
\begin{align}
\vv'&=\vv- [(\vv-\vv_*) \cdot \vomega]\vomega \label{v'1}\\
\vv'_* &= \vv_* + [(\vv-\vv_*)\cdot \vomega]\vomega \label{v'2}
\end{align}
for the outgoing velocities of a collision between two balls with incoming velocities $\vv$ and $\vv_*$ and direction $\vomega$ between the centers of the two balls.
The Boltzmann equation is considered for $\vv\in\RRR^3$ and $\vq\in\Lambda$ along with a boundary condition representing that balls hitting the boundary of $\Lambda$ get reflected there. A function $f(\vq,\vv)$ is a stationary solution if and only if it is independent of $\vq$ and a Maxwellian (i.e., Gaussian) in $\vv$---that is, if and only if it represents thermal equilibrium. Correspondingly, non-equilibrium macro states correspond to any density function $f$ that is not a global (i.e., $\vq$-independent) Maxwellian.

The entropy turns out to be (up to addition of the constant $\kk N$ and terms of lower order)
\be
S(x) = -\kk \sum_i \vol(C_i) \, Nf_i \log \bigl[Nf_i\bigr]\,.
\ee
In the limit of small $C_i$, this becomes \eqref{-Hdef}, i.e., $-\kk H$ in terms of the $H$ functional. Boltzmann further proved the $H$-theorem, which asserts that for any solution $f_t$ of the Boltzmann equation,
\be
\frac{dH}{dt} \leq 0\,,
\ee
with equality only if $f$ is a local Maxwellian. The $H$-theorem amounts to a derivation of the second law relative to the partition $\{\Gamma_\nu\}$ under consideration.

\subsection{Rigorous Result}

Some authors suspected that Boltzmann's vision, and the Boltzmann equation in particular, was not valid. For example, \citet[\S 33, p.~142]{Khi} complained about individualist accounts of entropy:
\begin{quotation}
All existing attempts to give a general proof of this postulate [i.e., $S=\kk\log W$,] must be considered as an aggregate of logical and mathematical errors superimposed on a general confusion in the definition of the basic quantities.
\end{quotation}
But actually, the Boltzmann equation (and with it the increase of entropy) is rigorously valid for most phase points in $\Gamma_\nu$, at least for a short time, as proved by Lanford \citeyearpar{Lan75,Lan76}. Here is a summary statement of Lanford's theorem, leaving out some technical details: 

\begin{thm}\label{thm:Lanford}
Let $\overline{t}>0$ and $\lambda>0$ (the constant in the Boltzmann equation) be constants.
For a very large number $N$ of billiard balls of (very small) radius $a$ with $4Na^2=\lambda$, for every $0\leq t < \tfrac{1}{5} \overline{t}$, for any nice density $f_0$ in $\X_1= \Lambda \times \RRR^3$ with mean free time $\geq \overline{t}$, and for a coarse graining of $\X_1$ into cells $C_i$ that are small but not too small, most phase points $X$ with empirical distribution $f_0$ (relative to $\{C_i\}$ and within small tolerances) evolve in such a way that the empirical distribution of $X(t)$ (relative to $\{C_i\}$) is close to $f_t$, where $t'\mapsto f_{t'}$ is the solution of the Boltzmann equation with initial datum $f_0$. 
\end{thm}

It is believed but not proven that the Boltzmann equation is valid for a much longer duration, maybe of the order of recurrence times. The method of proof fails after $\tfrac{1}{5}\overline{t}$, but it does not give reason to think that the actual behavior changes at $\tfrac{1}{5}\overline{t}$.

Where is the famous \emph{Stosszahlansatz}, or hypothesis of molecular chaos, in this discussion? This hypothesis was stated by Boltzmann as specifying the approximate number of collisions with parameter $\vomega$ between particles from cells $C_i$ and $C_j$ within a small time interval. In our discussion it is hidden in the assumption that the initial phase point $X(0)$ be \emph{typical} in $\Gamma_\nu$: Both Theorem~\ref{thm:Lanford} and the wording \eqref{2ndlaw} of the second law talk merely about \emph{most} phase points in $\Gamma_\nu$, and for most phase points in $\Gamma_\nu$ ($\nu=f$) it is presumably true that the number of upcoming collisions is, within reasonable tolerances, given by the hypothesis of molecular chaos, not just at the initial time but at all 
relevant times. We will discuss molecular chaos further in Section~\ref{sec:PH}.

\subsection{Empirical vs.\ Marginal Distribution}

Many mathematicians (e.g., \citet{Kac,CIP94} but also \citet{Tol38}) considered the Boltzmann equation in a somewhat different context with a different meaning, viz., with $f$ not the \emph{empirical} distribution but the \emph{marginal} distribution. This means the following. 
\begin{itemize}
\item The empirical distribution, for a given phase point $X=(\vq_1,\vv_1,\ldots,\vq_N,\vv_N)\in\X$, is the distribution on $\X_1$ with density
\be\label{fempdef}
f_\mathrm{emp}^X(\vq,\vv) = \frac{1}{N}\sum_{j=1}^N \delta^3(\vq_j-\vq) \, \delta^3(\vv_j-\vv) \,.
\ee
As such, it is not a continuous distribution but becomes roughly continuous-looking only after coarse graining with cells $C_i$ in $\X_1$ that are not too small (so that the occupation numbers are large enough), and it becomes a really continous distribution only after taking a limit in which the cells shrink to size 0 while $N\to\infty$ fast enough for the occupation numbers to become very large.
\item The marginal distribution starts from a distribution $\rho$ on phase space $\X$ and is obtained by integrating over the positions and velocities of $N-1$ particles (and perhaps averaging over the number $i$ of the particle not integrated out, if $\rho$ was not permutation invariant to begin with). The marginal distribution can also be thought of as the average of the (exact) empirical distribution: the empirical distribution $f_\mathrm{emp}^X$ associated with $X\in\X$ becomes a continuous function when $X$ is averaged over using a continuous $\rho$.
\end{itemize}

For example, \citet[first page]{Kac} wrote: 
\begin{quotation}
$f(r,v)dr \, dv$ is the average number of molecules in $dr\, dv$,
\end{quotation}
whereas 
\citet[\S 3 p.~36]{Bol1898} wrote:
\begin{quotation}
let $f(\xi,\eta,\zeta,t) d\xi\, d\eta \, d\zeta$ [\ldots] be the number of $m$-molecules whose velocity components in the three coordinate directions lie between the limits $\xi$ and $\xi+d\xi$, $\eta$ and $\eta +d\eta$, $\zeta$ and $\zeta +d \zeta$[.]
\end{quotation}
Note that Kac wrote ``average number'' and Boltzmann wrote ``number'': For Kac, $f$ was the marginal and for Boltzmann the empirical distribution.\footnote{That is also why Boltzmann normalized $f$ so that $\int_{\X_1} f =N$, not $\int_{\X_1}f=1$: a marginal of a probability distribution would automatically be normalized to 1, not $N$, but if $f$ means the empirical density then it is natural to take it to mean the density of particles in $\X_1$, which is normalized to $N$, not 1.}

Of course, the (coarse-grained) empirical distribution is a function of the phase point $X$, and so is any functional of it, such as $H$; thus, the empirical distribution can serve the role of the macro state $\nu$, and $H$ that of Boltzmann entropy. This is not possible for the marginal distribution. 

So why would anybody want the marginal distribution? Kac aimed at a rigorous derivation of the Boltzmann equation in whatever context long before Lanford's theorem, and saw better chances for a rigorous proof if he assumed collisions to occur at random times at exactly the rate given by Boltzmann's hypothesis of molecular chaos. This setup replaces the time evolution in phase space (or rather, since Kac dropped the positions, in $3N$-dimensional velocity space) by a stochastic process, in fact Markov jump process. (By the way, as a consequence, any density $\rho$ on $3N$-space tends to get wider over time, and its Gibbs entropy increases, contrary to the Hamiltonian evolution.) So the mathematician's aim of finding statements that are easier to prove leads in a different direction than the aim of discussing the mechanism of entropy increase in nature.

Another thought that may lead authors to the marginal distribution is that $f(\vq,\vv)$ $d^3\vq$ $d^3\vv$ certainly cannot be an integer but must be an infinitesimal, so it cannot be the number of particles in $d^3\vq\, d^3\vv$ but must be the \emph{average} number of particles. Of course, this thought neglects the idea that as long as $N$ is finite, also the cells $C_i$ should be kept of finite size and not too small, and the correct statement is that $f_i \, \vol C_i$ is the number of particles in $C_i$ (or, depending on the normalization of $f$, $N^{-1}$ times the number of particles); when followers of Boltzmann express the volume of $C_i$ as $d^3\vq\, d^3\vv$, they merely express themselves loosely.

\subsection{The Past Hypothesis}
\label{sec:PH}

Lanford's theorem has implications also for negative $t$: For most phase points in $\Gamma_f$, the Boltzmann equation also applies in the other time direction, so that entropy increases in both time directions! (See Figure~\ref{fig:2directions}.) That is, before time 0 the Boltzmann equation applies with the sign of $t$ reversed.

\begin{figure}[h]
\begin{center}
\begin{tikzpicture}
  \draw [thick, domain=0:6, samples=50] plot (\x,{2*(1-exp(-\x/2))});
  \draw [thick, domain=-6:0, samples=50] plot (\x,{2*(1-exp(\x))});
  \draw (-6,-0.5) -- (6,-0.5);
  \node at (5.8,-0.8) {$t$};
  \draw (-5,-1) -- (-5,2.5);
  \node at (-5.3,2.3) {$S$};
  \draw (0,-0.5) -- (0,-0.7);
  \node at (0,-0.9) {$0$};
\end{tikzpicture}
\end{center}
\caption{For most phase points in $\Gamma_\nu$, $\nu\neq \eq$, entropy increases in both time directions, albeit not necessarily at the same rate.}
\label{fig:2directions}
\end{figure}
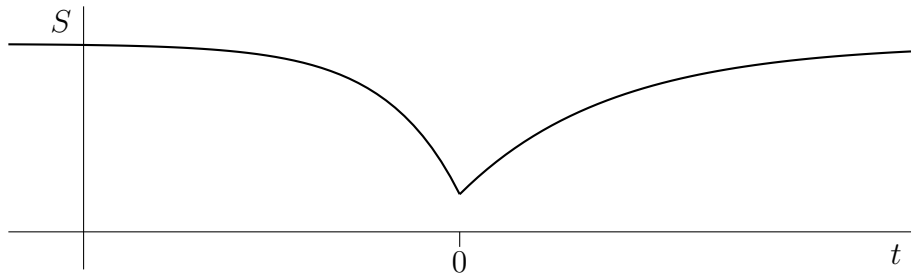

It is generally the case, not just for a dilute gas of hard spheres, that entropy increases in both time directions for most phase points in $\Gamma_\nu$. This fact leads to the following worry: If Lanford's theorem, or the statement \eqref{2ndlaw}, persuaded us to expect that entropy increases after the time we chose to call $t=0$, should it not persuade us just as much to expect that entropy decreases before $t=0$? But this decrease does not happen. Does that mean we were unjustified in expecting from Lanford's theorem or \eqref{2ndlaw} that entropy increases after $t=0$? 

Here is a variant of this worry in terms of explanation. Statement \eqref{2ndlaw} may suggest that the explanation for the increase of entropy after $t=0$ is that this happens for typical phase points $X(0)$. But if we know that entropy was lower before $t=0$ than at $t=0$, then $X(0)$ was not typical. This undercuts the explanation considered for the behavior after $t=0$.

Here is the resolution of this worry. The assumption really made by Boltzmann's followers is not that the phase point $X(0)$ at the beginning of an experiment is typical in $\Gamma_\nu$ but the following: 
\be\label{PH}
\mbox{\begin{minipage}{0.85\textwidth}
{\bf Past hypothesis.} The phase point of the universe at the initial time $T_0$ of the universe (presumably the big bang) is typical in its macro set $\Gamma_{\nu_0}$, where $\nu_0$ has very low entropy.
\end{minipage}}
\ee
Here, ``typical'' means that in relevant ways it behaves like most points in $\Gamma_{\nu_0}$. The ways relevant here are features of the macro history of the universe shared by most points in $\Gamma_{\nu_0}$.

Given that entropy keeps increasing, the initial macro state must be one of extremely low entropy; one estimate \citep{Pen} yields $10^{123}$ Joule/Kelvin less than the thermal equilibrium entropy at the same energy; thus, $\Gamma_{\nu_0}$ must have extremely small volume compared to the relevant energy shell. All the same, we do not know very clearly what $\nu_0$ actually is; one proposal is known as the Weyl curvature hypothesis \citep{Pen79}.

\bigskip

A related worry at this point may arise from the observation (for any macroscopic system or specifically a hard sphere gas as considered for the Boltzmann equation) that if the time evolution $\Phi_t$ of $\Gamma_\nu$ lies (say) in a macro set $\Gamma_{\nu'}$ of much greater volume, then phase points $X(t)$ coming from $X(0)\in\Gamma_\nu$ would be \emph{atypical} in $\Gamma_{\nu'}$. So if the prediction of entropy increase after $t$ was based on the assumption that $X(t)$ be \emph{typical} in $\Gamma_{\nu'}$, then it could not be applied to $X(0)\in \Gamma_\nu$. So why should entropy still increase at $t>0$? Because Lanford's theorem says so---at least until $\overline{t}/5$. But even after that time, it is plausible that the Boltzmann equation continues to be valid (and therefore entropy continues to increase) because it is plausible that the number of upcoming collisions of each type agrees with the value specified by the hypothesis of molecular chaos \citep[see, e.g.,][]{L07}. That is because it is plausible that, for typical $X(0)\in\Gamma_\nu$, $X(t)$ contains very special correlations in the exact positions and velocities of all particles concerning the collisions before $t$, but not concerning those after $t$. Likewise, we would expect for a general macroscopic system, unless the dynamics (as given by the Hamiltonian) is specially contrived for the partition $\{\Gamma_\nu\}$, that $X(t)\in\Gamma_{\nu'}$ coming from a typical $X(0)\in\Gamma_\nu$ behaves towards the future, but of course not towards the past, like a typical point in $\Gamma_{\nu'}$, as stated by the second law \eqref{2ndlaw}. The whole reasoning does not change significantly if $\Phi_t(\Gamma_\nu)$ is distributed over \emph{several} $\Gamma_{\nu_1'},\ldots,\Gamma_{\nu_\ell'}$ of much greater volume, instead of being contained in one $\Gamma_{\nu'}$.

Putting this consideration together with the past hypothesis, we are led to expect that the Boltzmann entropy of the universe keeps increasing (except for occasional insignificant entropy valleys) to this day, and further until the universe reaches thermal equilibrium. As a consequence for our present-day experiments:
\be\label{fact2}
\mbox{\begin{minipage}{0.85\textwidth}
{\it Development conjecture.} Given the past hypothesis, an isolated system that, at a time $t_0$ before thermal equilibrium of the universe, has macro state $\nu$ appears macroscopically in the future, but not the past, of $t_0$ like a system that at time $t_0$ is in a typical micro state compatible with $\nu$.
\end{minipage}}
\ee
This statement follows from Lanford's theorem for times up to $\overline{t}/5$, but otherwise has not been proven mathematically; it summarizes the presumed implications of the past hypothesis (i.e., of the low entropy initial state of the universe) to applications of statistical mechanics. For a dilute gas, it predicts that its macro state will evolve according to the Boltzmann equation in the future, but not the past, of $t_0$, as long as it is isolated. It also predicts that heat will not flow from the cooler body to the hotter, and that a given macroscopic object will not spontaneously fly into the air although the laws of mechanics would allow that all the momenta of the thermal motion at some time all point upwards.

By the way, the development conjecture allows us to make sense of Tolman's \citeyearpar[\S 23, p.~60]{Tol38} ``hypothesis of equal a priori probabilities,'' which asserts  
\begin{quotation}
that the phase point for a given system is just as likely to be in one region of the phase space as in any other region of the same extent which corresponds equally well with what knowledge we do have as to the condition of the system.
\end{quotation}
That sounds like $X$ is always uniformly distributed over the $\Gamma_\nu$ containing $X$, but that statement is blatantly inconsistent, as it cannot be true at two times, given that $\Phi_t(\Gamma_\nu)\neq \Gamma_{\nu'}$ for any $\nu'$. But the subtly different statement \eqref{fact2} is consistent, and \eqref{fact2} is what Tolman should have written.

\bigskip

The past hypothesis brings out clearly that the origin of the thermodynamic arrow of time lies in a special property (of lying in $\Gamma_{\nu_0}$) of the physical state of the matter at the time $T_0$ of the big bang, and not in the observers' knowledge or their way of considering the world. The past hypothesis is the one crucial assumption we make in addition to the dynamical laws of classical mechanics. The past hypothesis may well have the status of a law of physics---not a dynamical law but a law selecting a set of admissible histories among the solutions of the dynamical laws. As \citet[p.~116]{Fey65} wrote:
\begin{quotation}
Therefore I think it is necessary to 
add to the physical laws the hypothesis that in the past the universe 
was more ordered, in the technical sense, than it is today---I think this is the additional 
statement that is needed to make sense, and to make an understanding 
of the irreversibility.
\end{quotation}

Making the past hypothesis explicit in the form \eqref{PH} or a similar one also enables us to understand the question whether the past hypothesis could be \emph{explained}. \citet[\S 90]{Bol1898} suggested tentatively that the explanation might be a giant fluctuation out of thermal equilibrium, an explanation later refuted by \citet{Edd31} and Feynman \citeyearpar{Fey65,Fey95}. (Another criticism of this explanation put forward by \citet[\S 35]{Pop76} is without merit.) Some explanations of the past hypothesis have actually been proposed in  highly idealized models \citep{Car,BKM13,BKM14,BKM15,GTZ16}; it remains to be seen whether they extend to more adequate cosmological models.

\subsection{Boltzmann Entropy Is Not Always $H$}
\label{sec:objections}

\citet{Jay65} wrote that
\begin{quotation}
the Boltzmann $H$ yields an ``entropy'' that is in error by a nonnegligible amount whenever interparticle forces affect thermodynamic properties.
\end{quotation}
It is correct that the $H$ functional represents (up to a factor $-\kk$) the Boltzmann entropy \eqref{SBdef} only for a gas of non-interacting or weakly interacting molecules. Let us briefly explain why; a more extensive discussion is given in \citep{GL04}. 

As pointed out in Section~\ref{sec:macrostates}, the coarse grained energy \eqref{M1} should be one of the macro variables, so that the partition into macro sets $\Gamma_\nu$ provides a partition of the energy shell. If interaction cannot be ignored, then the $H$ functional does 
not correspond to the Boltzmann entropy, since restriction to the energy 
shell is not taken into account by $H$. When interaction 
can be ignored there is only kinetic energy, so the Boltzmann 
macro states based on the empirical distribution alone determine the 
energy and hence the $H$ functional corresponds to the Boltzmann entropy.

\section{Gibbs Entropy as Fuzzy Boltzmann Entropy}
\label{sec:agree}

This section is about another connection between Gibbs and Boltzmann entropy that is not usually made explicit in the literature; it involves interpreting the Gibbs entropy in an individualist spirit as a variant of the Boltzmann entropy for ``fuzzy'' macro sets. That is, we now describe a kind of Gibbs entropy that is not an ensemblist entropy (as the Gibbs entropy usually is) but an individualist entropy, which is better. This possibility is not captured by any of the options (a)--(c) of Section~\ref{sec:xrho}.

By a fuzzy macro set we mean using functions $\gamma_\nu(x)\geq 0$ instead of sets $\Gamma_\nu$ as expressions of a macro state $\nu$: some phase points $x$ look a lot like $\nu$, others less so, and $\gamma_\nu(x)$ quantifies how much. The point here is to get rid of the sharp boundaries between the sets $\Gamma_\nu$ shown in Figure~\ref{fig:phasespace}, as the boundaries are artificial and somewhat arbitrary anyway. A partition into sets $\Gamma_\nu$ is still contained in the new framework as a special case by taking $\gamma_\nu$ to be the indicator function of $\Gamma_\nu$, $\gamma_\nu = 1_{\Gamma_\nu}$, but we now also allow continuous functions $\gamma_\nu$. One advantage may be to obtain simpler expressions for the characterization of $\nu$ since we avoid drawing boundaries. The condition that the $\Gamma_\nu$ form a partition of $\X$ can be replaced by the condition that
\be\label{sumgammanu}
\sum_\nu \gamma_\nu(x) = 1~~~~\forall x\in\X.
\ee
Another advantage of this framework is that we can allow without further ado that $\nu$ is a continuous variable, by replacing \eqref{sumgammanu} with its continuum version
\be\label{intgammanu}
\int d\nu \, \gamma_\nu(x) = 1~~~~\forall x\in\X.
\ee
It will sometimes be desirable to normalize the function $\gamma_\nu$ so its integral becomes 1; we write $\gamma^1_\nu$ for the normalized function,
\be
\gamma^1_\nu(x) = \frac{\gamma_\nu(x)}{\|\gamma_\nu\|_1} ~~~\text{with}~~~\|\gamma_\nu\|_1 :=\int_\X \!\! dx' \, \gamma_\nu(x')\,.
\ee
So what would be the appropriate generalization of the Boltzmann entropy to a fuzzy macro state? It should be $\kk$ times the log of the volume over which $\gamma_\nu^1$ is effectively distributed---in other words, the Gibbs entropy,
\be
S_\B(\nu) := S_\G(\gamma^1_\nu)\,.
\ee

Now fix a phase point $x$. Since $x$ is now not uniquely associated with a macro state, it is not clear what $S_\B(x)$ should be. In view of \eqref{intgammanu}, one might define $S_\B(x)$ to be the average $\int d\nu \, S_\B(\nu) \, \gamma_\nu(x)$. Be that as it may, if the choice of macro states $\gamma_\nu$ is reasonable, one would expect that different $\nu$s for which $\gamma_\nu(x)$ is significantly non-zero have similar values of $S_\B(\nu)$, except perhaps for a small set of exceptional $x$s.

Another advantage of fuzzy macro states is that they sometimes factorize in a more convenient way. Here is an example. To begin with, sometimes, when we are only interested in thermal equilibrium macro states, we may want to drop non-equilibrium macro states and replace the equilibrium set $\Gamma_\eq$ in an energy shell $\X_{(E-\Delta E,E]}$ by the full energy shell, thereby accepting that we attribute wildly inappropriate $\nu$s (and $S_\B$s) to a few $x$s. As a second step, the canonical distribution
\be
\rho_\can(\beta,x) = \frac{1}{Z} e^{-\beta H(x)}
\ee
is strongly concentrated in a very narrow range of energies, a fact well known as ``equivalence of ensembles'' (between the micro-canonical and the canonical ensemble). Let us take $\nu=\beta$ and $\gamma^1_\nu(x)=\rho_\can(\beta,x)$ as a continuous family of fuzzy macro states. As a third step, consider a system consisting of two non-interacting subsystems, $\sys=\sys_1 \cup \sys_2$, so $\X=\X_1 \times \X_2$, $x=(x_1,x_2)$ and $H(x) = H_1(x_1) + H_2(x_2)$. Then $\gamma^1_\beta$ factorizes,
\be
\gamma^1_\beta(x) = \gamma^1_{1,\beta}(x_1) \, \gamma^1_{2,\beta}(x_2)\,,
\ee
whereas the energy shells do not,
\be\label{neq}
\X_{(E-\Delta E,E]} \neq 
\X_{1,(E_1-\Delta E_1,E_1]} \times \X_{2,(E_2-\Delta E_2,E_2]}\,,
\ee
because a prescribed total energy $E$ can be obtained as a sum $E_1+E_2$ for very different values of $E_1$ through suitable choice of $E_2$, corresponding to different macro states for $\sys_1$ and $\sys_2$. One particular splitting $E=E_1+E_2$ will have the overwhelming majority of phase space volume in $\X_{(E-\Delta E,E]}$; it is the splitting that maximizes $S_1(E_1)+S_2(E_2)$ under the constraint $E_1+E_2=E$, and at the same time the one corresponding to the same $\beta$ value, i.e., with $E_i$ the expectation of $H_i(x_i)$ under the distribution $\gamma^1_{i,\beta}$ ($i=1,2$). So equality in \eqref{neq} fails only by a small amount, in the sense that the symmetric difference set between the left and right-hand sides has small volume compared to the sets themselves. In short, in this situation, the use of sets $\Gamma_\nu$ forces us to consider certain non-equilibrium macro states, and if we prefer to introduce only equilibrium mecro states, then the use of fuzzy macro states $\gamma_\nu$ is convenient.

\bigskip

The quantum analog of fuzzy macro states consists of, instead of the orthogonal decomposition \eqref{decomp} of Hilbert space, a POVM (positive-operator-valued measure) $\hat G$ on the set $M$ of macro states. That is, we replace the projection $\hat P_\nu$ to the macro space $\Hilbert_\nu$ by a positive operator $\hat G_\nu$ with spectrum in $[0,1]$ such that $\sum_\nu \hat G_\nu = \hat I$, where $\hat I$ is the identity operator on $\Hilbert$. The eigenvalues of $\hat G_\nu$ would then express how much the corresponding eigendirection looks like $\nu$. Set $\hat G^1_\nu := (\tr \hat G_\nu)^{-1} \hat G_\nu$; then $\hat G^1_\nu$ is a density matrix, and the corresponding entropy value would be
\be
S_\qB(\nu) := S_\vN(\hat G^1_\nu)\,.
\ee
As before when using the $\Hilbert_\nu$, a quantum state $\psi$ may be associated with several very different $\nu$s. As discussed around \eqref{dontoverlap}, this problem presumably disappears in Bohmian mechanics and collapse theories.

\section{The Status of Ensembles}
\label{sec:status}

If we use the Boltzmann entropy, then the question of what the $\rho$ in the Gibbs entropy means does not come up. But a closely related question remains: What is the meaning of the Gibbs ensembles (the micro-canonical, the canonical, the grand-canonical) in Boltzmann's ``individualist'' approach? This is the topic of the present section \citep[see also][]{Gol19}.

By way of introduction to this section, we can mention that Wallace, an ensemblist, feels the force of arguments against subjective entropy but thinks that there is no alternative. He wrote \citep[Sec.~10]{Wal18}:
\begin{quotation}
It will be objected by (e.g.)\ Albert and Callender that when we say ``my coffee will almost certainly cool to room temperature if I leave it'' we are saying something objective about the world, not something about my beliefs. I agree, as it happens; that just tells us that the probabilities of statistical mechanics cannot be interpreted epistemically. And then, of course, it is a mystery how they can be interpreted, given that the underlying dynamics is deterministic [\ldots]. But (on pain of rejecting a huge amount of solid empirical science [\ldots]) some such interpretation must be available.
\end{quotation}
The considerations of this section may be helpful here.

It is natural to call any measure (or density function $\rho$) that is normalized to 1 a ``probability distribution.'' But now we need to distinguish more carefully between different roles that such a measure may play:
\begin{itemize}
\item[(i)] frequency in repeated preparation: The outcome of an experiment is unpredictable, but will follow the distribution $\rho$ if we repeat the experiment over and over. (This kind of probability could also be called \emph{genuine probability} or \emph{probability in the narrow sense}.)
\item[(ii)] degree of belonging: As in Section~\ref{sec:agree}, $\rho$ could represent a fuzzy set; then $\rho(x)$ indicates how strongly $x$ belongs to this fuzzy set.
\item[(iiii)] typicality: This is how the flat distribution over a certain $\Gamma_\nu$ enters the past hypothesis \eqref{PH} and the development conjecture \eqref{fact2}.
\end{itemize}

\subsection{Typicality}
\label{sec:typicality}

A bit more elucidation of the concept of typicality may be useful here, as the difference between typicality and genuine probability is subtle. A feature or behavior is said to be \emph{typical} in a set $S$ if it occurs for most (i.e., for the overwhelming majority of) elements of $S$. Let us elaborate on this by means of an example. The digits of $\pi$, 314159265358979\ldots, look very much like a random sequence although there is nothing random about the number $\pi$, as it is uniquely defined and thus fully determined. Also, there is no way of ``repeating the experiment'' that would yield a different sequence of digits; we can only do statistics about the digits in this one sequence. We can set up reasonable criteria for whether a given sequence (finite or infinite) ``looks random,'' such as whether the relative frequency of each digits is, within suitable tolerances, 1/10, and that of each $k$-digit subsequence $1/10^k$ (for all $k$ much smaller than the length of the sequence). Then the digits of $\pi$ will presumably pass the criteria, thereby illustrating the fine but relevant distinction between ``looking random'' and ``being random.'' 

Looking random is an instance of typical behavior, as most sequences of digits of given length look random and most numbers between 3 and 4 have random-looking decimal expansions---``most'' relative to the uniform distribution, also known as Lebesgue measure. So with respect to a certain behavior (say, the frequencies of $k$-digit subsequences), $\pi$ presumably is like most numbers.

Similarly, with respect to another sort of behavior (say, the future macroscopic history of the universe), the initial phase point of the universe is presumably like most points in $\Gamma_{\nu_0}$, as demanded by the past hypothesis \eqref{PH}; and according to the development conjecture \eqref{fact2}, the phase point of a closed system now is like most points in its macro set.

In contrast, when considering a random experiment in probability theory, we usually imagine that we can repeat the experiment, with relative frequencies in agreement with the distribution $\rho$; that is clearly impossible for the universe as a whole (or, if it were possible, irrelevant, because we are not concerned with what happens in other universes). But even if an experiment, such as the preparation of a macroscopic body (say, a gas), can be repeated a hundred or a thousand times, the frequency distribution of the sample points (i.e., the empirical distribution \eqref{fempdef}) consists of a mere $10^3$ delta peaks and thus is far from resembling a continuous distribution in a phase space of $10^{23}$ dimensions, so we do not have much basis for claiming that it is essentially the same as the micro-canonical or the canonical or any other common ensemble $\rho$. On the other hand, even a single sample point could meaningfully be said to be typical with respect to a certain high-dimensional distribution $\rho$ (and a certain concept of which kinds of ``behavior'' are being considered). What we are getting at is that Gibbs's ensembles are best understood as measures \emph{of typicality, not of genuine probability}.

For example, there is a subtle difference between the Maxwell-Boltzmann distribution and the canonical distribution. The former is the empirical distribution of (say) $10^{23}$ points in 6-dimensional 1-particle phase space $\X_1$, while the latter is a distribution of a single or a few points in high-dimensional phase space. The former is genuine probability, the latter typicality. This contrast is also related to the contrast between empirical and marginal distribution: On the one hand, the canonical distribution of a system $\sys_1$ is often derived as a marginal distribution of the micro-canonical distribution of an even bigger system $\sys_1\cup \sys_2$, and a marginal of a typicality distribution is still a typicality distribution. On the other hand, the Maxwell-Boltzmann distribution is the empirical distribution that arises from most phase points relative to the micro-canonical distribution.  If $\rho$-most phase points lead to the empirical distribution $f$, then that explains why we can regard $f$ as a \emph{genuine probability} distribution, in the sense of being objective.

\subsection{Equivalence of Ensembles and Typicality}
\label{sec:eqens}

It is not surprising that typicality and genuine probability can easily be conflated, and that various conundrums can arise from that. For example, the ``mentaculus'' of Albert and Loewer \citep{Alb15,Loe19} is a view based on understanding the past hypothesis \eqref{PH} with genuine probability instead of typicality. That is, in this view one considers the uniform distribution over the initial macro set $\Gamma_{\nu_0}$ and regards it as defining ``the'' probability distribution of the micro history $t\mapsto X(t)$ of the universe, in particular as assigning a precisely defined value in $[0,1]$ as the probability to every conceivable event, for example the event that intelligent life forms evolve on Earth and travel to the moon. 

In contrast, the typicality approach used in \eqref{PH} does not require or claim that the initial phase point $X(T_0)$ of the universe is random, merely that it looks random (just as $\pi$ merely looks random). Furthermore, the typicality view does not require a uniquely selected measure of typicality, just as $\pi$ is presumably typical relative to several measures: the uniform distribution on the interval $[3,4]$ or on $[0,10]$, or some Gaussian distribution. We elaborate on this aspect in the subsequent paragraphs, again through contrast with Albert and Loewer.

Loewer once expressed (personal communication 2018) that different measures $\rho_1\neq \rho_2$ for the initial phase point of the universe would necessarily lead to different observable consequences. That is actually not right; let us explain. To begin with, one may think that $\rho_1$ and $\rho_2$ must lead to different observable consequences if one \emph{defines} (following the mentaculus) the probability of any event $E$ at time $t$ as $P(E)=\int_{\Phi_t^{-1}(E)}\rho_i$. However, a definition alone does not ensure that the quantity $P(E)$ has empirical significance; the question arises whether observers inside the universe can determine $P(E)$ from their observations. More basically, we need to ask whether observers can determine from their observations which of $\rho_1$ and $\rho_2$ was used as the initial distribution. To this end, let us imagine that the initial phase point $X(T_0)$ of the universe gets chosen truly randomly with distribution either $\rho_1$ or $\rho_2$. Now the same point $X(T_0)$ may be compatible with both distributions $\rho_1,\rho_2$ if they overlap appropriately. Since all observations must be made in the one history arising from $X(T_0)$, and since the universe cannot be re-run with an independent initial phase point, it is impossible to decide empirically whether $\rho_1$ or $\rho_2$ was the distribution according to which $X(T_0)$ was chosen. Further thought in this direction shows that observers can only reliably distinguish between $\rho_1$ and $\rho_2$ if they can observe an event that has probability near 1 for $\rho_1$ and near 0 for $\rho_2$ or vice versa; that is, they can only determine the typicality class of the initial distribution. (As a consequence for the mentaculus, its main work, which is to efficiently convey hard core facts about the world corresponding to actual patterns of events over space and time, is done via typicality, i.e., concerns only events with measure near 0 or near 1.)

In the typicality view there does not have to be a fact about which distribution $\rho$ is ``the right'' distribution of the initial phase point of the universe; that is because points typical relative to $\rho_1$ are also typical relative to $\rho_2$ and vice versa, if $\rho_1$ and $\rho_2$ are not too different. Put differently, in the typicality view there is a type of equivalence of ensembles, parallel to the well known fact that one can often replace a micro-canonical ensemble by a canonical one without observable consequences. There is room for different choices of $\rho$, and this fits well with the fact that $\Gamma_{\nu_0}$ has boundaries with some degree of arbitrariness.

\subsection{Ensemblist vs.\ Individualist Notions of Thermal Equilibrium}

The two notions of Gibbs and Boltzmann entropy are parallel to two notions of thermal equilibrium. In the view that we call the ensemblist view, a system is in thermal equilibrium if and only if its phase point $X$ is random with the appropriate distribution (such as the micro-canonical or the canonical distribution). In the individualist view, in contrast, a system is in thermal equilibrium (at a given energy) if and only if its phase point $X$ lies in a certain subset $\Gamma_{\eq}$ of phase space. 

The definition of this set may be quite complex; for example, according to the partition due to Boltzmann described in \eqref{occupation}, the set $\Gamma_{\eq}$ contains those phase points for which the relative number of particles in each cell $C_i$ of $\X_1$ agrees, up to tolerance $\Delta f$, with the content of $C_i$ under the Maxwell-Boltzmann distribution. In general, the task of selecting the set $\Gamma_{\eq}$ has similarities with that of selecting the set $S$ of those sequences of (large) length $L$ of decimal digits that ``look random.'' While a reasonable test of ``looks random'' can be designed (such as outlined in the first paragraph of Section~\ref{sec:typicality} above), different scientists may well come up with criteria that do not exactly agree, and again there is some arbitrariness in where exactly to ``draw the boundaries.'' Moreover, such a test or definition $S$ is often unnecessary for practical purposes, even for designing and testing pseudo-random number generators. For practical purposes, it usually suffices to specify the distribution (here, the uniform distribution over all sequences of length $L$) with the understanding that $S$ should contain, in a reasonable sense, the typical points relative to that distribution. Likewise, we usually never go to the trouble of actually selecting $\Gamma_{\eq}$ and the other $\Gamma_\nu$: it often suffices to imagine that they \emph{could} be selected. Specifically, for thermal equilibrium, it often suffices to specify the distribution (here, the uniform distribution over the energy shell, that is, the micro-canonical distribution) with the understanding that $\Gamma_\eq$ should contain, in a reasonable sense, the typical points relative to that distribution (i.e., points that look macroscopically the way most points in $\X_\mc$ do). 

That is why, also for the individualist, specifying an ensemble $\rho$ can be a convenient way of describing a thermal equilibrium state: It would basically specify the set $\Gamma_\eq$ as the set of those points $x$ that look macroscopically like $\rho$-most points. For the same reasons as in Section~\ref{sec:eqens}, this does not depend sensitively on $\rho$: points typical relative to the micro-canonical ensemble near energy $E$ are also typical relative to the canonical ensemble with mean energy $E$.

\bigskip

Let us turn to the ensemblist view. The basic problem with the ensemblist definition of thermal equilibrium is the same as with the Gibbs entropy: a system has an $X$ but not a $\rho$. To say the least, it remains open what should be meant by $\rho$. Is it subjective? But whether or not a system is in thermal equilibrium is not subjective. Does that mean that in distant places with no observers around, no system can be in thermal equilibrium? That does not sound right. 

What if $\rho$ corresponded to a repeated preparation procedure? That is problematical as well. Suppose such a procedure produced a random phase point $X(0)$ that is uniformly distributed in some set $M$; now let time evolve up to $t$; the resulting phase point $X(t)$ is uniformly distributed in $\Phi_t(M)$, the associated Gibbs entropy is constant, and so the distribution of $X(t)$ is far from thermal equilibrium. In other words, the system never approaches thermal equilibrium. That does not sound right. We will come back to this point in Section~\ref{sec:indep}.

Moreover, thermal equilibrium should have something to do with physical behavior, which depends on $X$. Thermal equilibrium should mean that the temperature, let us say the energy per degree of freedom, is constant (within tolerances) throughout the volume of the system, and that is a property of $X$. Thermal equilibrium should mean that the values of macro variables, which are functions on phase space, are constant over time, and that is true of some $X$ and not of others.

\subsection{Ergodicity and Mixing}
\label{sec:ergodic}

Ergodic and mixing are two kinds of chaoticity properties of a dynamical system $\Y$ such as an energy surface in phase space. A measure-preserving dynamical system in $\Y$ is called ergodic if almost every trajectory spends time, in the long run, in all regions of $\Y$ according to their volumes. Equivalently, a system is ergodic if time averages almost surely equal phase averages. A measure-preserving dynamical system in $\Y$ is called mixing if for all subsets $A,B$ of $\Y$, $\vol(A \cap \Phi_t(B))$ tends to $\vol(A) \, \vol(B)/\vol(\Y)$ as $t\to \infty$. In other words, mixing means that for large $t$, $\Phi_t(B)$ will be deformed into all regions of phase space according to their volumes.

It is sometimes claimed (e.g., \citet{Mac89}) that ergodicity or mixing are crucial for statistical mechanics. We would like to explain now why that is not so---with a few caveats; see also \citet{CCCV16}.

\begin{itemize}
\item Mixing would seem relevant to the ensemblist approach to thermal equilibrium because mixing implies that the uniform distribution over $\Phi_t(B)$ converges setwise to the uniform distribution (Liouville measure) over the energy surface $\Y$ (and every density converges weakly to a constant) \citep{Kry,Rue}. However, for every finite $t$, the Gibbs entropy has not increased, so some ensemblists are pushed to say that the entropy does not increase within finite time, but jumps at $t=\infty$ to the thermal equilibrium value. Thus, some ensemblists are pushed to regard thermal equilibrium as an idealization that never occurs in the real world. That does not seem like an attractive option.

\item The fact that the observed thermal equilibrium value $a_\eq$ of a macro variable $A$ coincides with its micro-canonical average $\langle A \rangle$ has sometimes been explained through the following reasoning \citep[e.g.,][\S~9]{Khi}: {\it Any macro measurement takes a time that is long compared to the time that collisions take or the time of free flight between collisions. Thus, it can be taken to be infinite on the micro time scale. Thus, the measured value is actually not the value of $A(X(t))$ at a particular $t$ but rather its time average. By ergodicity, the time average is equal to the phase average, QED.}

This reasoning is incorrect. In fact, ergodicity is neither necessary nor sufficient for $a_{\eq}=\langle A \rangle$. Not necessary because $A$ is nearly constant over the energy surface, so most phase points will yield a value close to $\langle A \rangle$ even if the motion is not ergodic. And not sufficient because the time needed for the phase point of an ergodic system to explore an energy surface is of the order $10^N$ years and thus much longer than the duration of the measurement. This point is also illustrated by the fact that in a macro system with non-uniform temperature, you can clearly measure unequal temperatures in different places with a thermometer faster than the temperature equilibrates.

\item Still, ergodicity and similar properties are not completely unrelated to thermodynamic behavior. Small interactions between different parts of the system are needed to drive the system towards a dominant macro set (i.e., towards thermal equilibrium), and the same small interactions often also have the consequence of making the dynamics ergodic. That is, ergodicity is neither cause nor consequence of thermodynamic behavior, but the two have a common cause.

\item In addition, and aside from thermodynamic behavior or entropy increase, mixing and similar properties play a role for the emergence of macroscopic randomness. For example, suppose a coin is tossed or a die is rolled; chaotic properties of the dynamics ensure that the outcome as a function of the initial condition varies over a very small scale on phase space, with the consequence that a probability density over the initial conditions that varies over a larger scale will lead to a uniform probability distribution of the outcome. (For more extensive discussion see \citep[Sec.~4.1]{DT}.)

\item Boltzmann himself considered ergodicity, but not for the purposes above. He used it as a convenient assumption for estimating the mean free path of a gas molecule \citep[\S 10]{Bol1898}.
\end{itemize}

\subsection{Remarks}

Of course, Boltzmannian individualists use probability, too. There is no problem with considering procedures that prepare random phase points, or with using Gibbs ensembles as measures of typicality. 

When arguing against the individualist view, \citet[Sec.~4]{Wal18} mentioned as an example that transport coefficients (such as thermal conductivity) can be computed using the two-time correlation function
\be
C_{ij}(t) = \langle X_i(t) X_j(0)\rangle\,,
\ee
where $\langle \cdot \rangle$ means averaging while taking $X(0)$ as random with micro-canonical distribution. Wallace wrote:
\begin{quotation}
since $C(t)$ is an explicitly probabilistic quantity, it is not even defined on the Boltzmannian approach.
\end{quotation}
Actually, that is not correct. The individualist will be happy as soon as it is shown that for most phase points in $\X_\mc$, the rate of heat conduction is practically constant and can be computed from $C(t)$ in the way considered.

Another objection of Wallace's \citeyearpar[Sec.~5]{Wal18} concerns a system $\sys$ for which (say) two macro sets, $\Gamma_{\nu_1}$ and $\Gamma_{\nu_2}$, each comprise nearly half of the volume of the energy shell, so that the thermal equilibrium state is replaced by two macro states $\nu_1$ and $\nu_2$, as can occur with a ferromagnet (see Footnote~\ref{fn:ferro} in Section~\ref{sec:macrostates}). Wallace asked how an individualist can explain the empirical observation that the system is, in the long run, equally likely to reach $\nu_1$ or $\nu_2$. To begin with, if a preparation procedure leads to a random phase point with distribution $\rho$, then $\int_{\Phi_t^{-1}(\Gamma_{\nu_1})}\rho$ is the probability of $\sys$ being in $\nu_1$ at time $t$. As mentioned in Section~\ref{sec:ergodic} (penultimate bullet), a broad $\rho$ will lead to approximately equal probabilities for $\nu_1$ and $\nu_2$. This leads to the question why practical procedures lead to broad $\rho$s, and that comes from typicality as expressed in the development conjecture. Put differently, for a large number of identical ferromagnets it is typical that about half of them are in $\nu_1$ and about half of them in $\nu_2$.

\section{Second Law for the Gibbs Entropy?}
\label{sec:indep}

To the observation that the Gibbs entropy is constant under the Liouville evolution \eqref{Liouville}, different authors have reacted in different ways. We describe some reactions along with some comments of ours. (Of course, our basic reservations about the Gibbs entropy are those we expressed in Section~\ref{sec:subjective}.) 
\begin{itemize}
\item \citet{Khi} thought it was simply wrong to say that entropy increases in a closed system. He thought that the second law should merely assert that when a system is in thermal equilibrium and a constraint is lifted, then the new thermal equilibrium state has higher entropy than the previous one.

With this attitude one precludes the explanation of many phenomena. 

\item \citet{Rue} expressed the view that entropy stays constant at every finite time but increases at $t=\infty$ due to mixing (see Section~\ref{sec:ergodic} above).

It remains unclear how this view can be applied to any real-world situation, in which all times are finite.

\item Some authors (e.g., \citet{Mac89}) suggested considering a system $\sys_1$ perturbed by external noise (in the same spirit \citet{Kem39}). In the framework of closed Hamiltonian systems of finite size, this situation could be represented by interaction with another system $\sys_2$. It can be argued that if $\sys_2$ is large and if initially $\rho= \rho_1(x_1) \, \rho_2(x_2)$ with $\rho_2$ a thermal equilibrium distribution, then the marginal distribution
\be
\rho_1^\mathrm{marg}(x_1,t) := \int_{\X_2} dx_2 \, \rho(x_1,x_2,t)
\ee
will follow an evolution for which $S_\G(\rho_1^\mathrm{marg})$ increases with time, and one might want to regard the latter quantity as the entropy of $\sys_1$.

This attitude has some undesirable consequences: First, if the entropy of $\sys_1$ increases and that of $\sys_2$ increases then, since that of $\sys_1\cup \sys_2$ is constant, entropy is not additive. Second, since it is not true of \emph{all} $\rho$ on $\X_1\times \X_2$ that $S_\G(\rho_1^\mathrm{marg})$ increases with time, the question arises for which $\rho$ it will increase, and why real-world systems should be of that type.

\item Some authors \citep[e.g.,][\S~51]{Wal18,Tol38} have considered a partition of (an energy shell in) phase space $\X$, such as the $\Gamma_\nu$, and suggested coarse graining any $\rho$ according this partition.\footnote{\citet{Tol38} had in mind a partition into cells of equal size, but we will simply use the $\Gamma_\nu$ in the following.} That is, the coarse-grained $\rho$ is the function $P\rho$ whose value at $x\in \Gamma_\nu$ is the average of $\rho$ over $\Gamma_\nu$,
\be
P\rho (x) = \frac{1}{\vol \Gamma_\nu }\int_{\Gamma_\nu} \!\! dy \, \rho(y) \,.
\ee
$P$ is a projection to the subspace of $L^1(\X)$ of functions that are constant on every $\Gamma_\nu$. 
The suggestion is that thermodynamic entropy $S_t$ is given by 
\be\label{SGPrhot}
S_t:=S_\G(P\rho_t)\,,
\ee
with $\rho_t$ evolving according to the Liouville equation \eqref{Liouville}, and that this $S_t$ tends to increase with time. 

It is plausible that $S_t$ indeed tends to increase, provided that the time evolution has good mixing properties (so that balls in phase space quickly become ``spaghetti'') and that the scale of variation of $\rho_0$ is not small; see also \citet{FPPRV07}.
(The argument for the increase of \eqref{SGPrhot} given by \citet{Tol38} is without merit. After correctly showing (p.~169) that, for every distribution $\rho$,
\be\label{SGP}
S_\G(P\rho) \geq S_\G(\rho)\,,
\ee
Tolman assumes that $\rho_0$ is uniform in each cell, i.e., $P\rho_0=\rho_0$, and then points out that $S_\G(P\rho_t)\geq S_\G(P\rho_0)$ for all $t$. This is true, but it is merely a trivial consequence of \eqref{SGP} and \eqref{SGt} owed to the special setup for which $P\rho_0=\rho_0$, and not a general statement of entropy increase, as nothing follows about whether $S_\G(P\rho_{t_2}) > S_\G(P \rho_{t_1})$ for arbitrary $t_2>t_1$.)

\item Some authors have argued that the Liouville evolution \eqref{Liouville} is not appropriate for evolving $\rho$.

That depends on what $\rho$ means. If we imagine a preparation procedure that produces a random phase point $X(0)$ with probability distribution $\rho_0$, then letting the system evolve up to time $t$ will lead to a random phase point $X(t)$ with distribution $\rho_t$ given by the Liouville evolution. However, if we think of $\rho$ as representing an observer's belief or knowledge, the situation may be different. In the last two bullet points, we consider some alternative ways of evolution.

\item Let $L$ be the linear differential operator such that $L\rho$ is the right-hand side of the Liouville equation \eqref{Liouville}; then the Liouville evolution is $\rho_t=e^{Lt}\rho_0$. The proposal \eqref{SGPrhot} could be re-expressed by saying that when using the Gibbs entropy, $S_t=S_\G(\rho_t)$, the appropriate evolution for $\rho$ should be
\be\label{PeLt}
\rho_t = P e^{Lt}\rho_0~~~~\text{for } t>0 \,.
\ee

\item Instead of \eqref{PeLt}, some authors \citep{Mac89} suggested that
\be\label{PLPLPL}
\rho_t = \lim_{n\to\infty} (Pe^{Lt/n})^n \rho_0~~~~\text{for } t>0  \,.
\ee
This evolution corresponds to continually coarse graining $\rho$ after each infinitesimal time step of Liouville evolution.

Then the Gibbs entropy $S_\G(\rho_t)$ is indeed non-decreasing with $t$, in fact without exception, as a consequence of \eqref{SGP} and \eqref{SGt}. If the macro evolution is deterministic, then it should agree with the Boltzmann entropy. A stochastic macro evolution would yield additional contributions to the Gibbs entropy not present in the Boltzmann entropy (contributions arising from spread over many macro sets), but these contributions are presumably tiny. Furthermore, it seems not implausible that $\rho_t$ is often close to $Pe^{Lt}P\rho_0$, which implies that $S_\G(\rho_t)$ is close to the expectation of $S_\B(X_t)$ with $X_0$ random with distribution $P\rho_0$. However, as discussed in Section~\ref{sec:wrongvalues}, one should distinguish between the entropy and the average entropy, so $S_\G(\rho_t)$ is sometimes not quite the right quantity.
\end{itemize}

\section{Further Aspects of the Quantum Case}
\label{sec:qm}

Much of what we have discussed in the classical setting re-appears mutatis mutandis in the quantum case, and we will focus now on the differences and specialties of the quantum case. We are considering a quantum system of a macroscopic number $N$ (say, $>10^{23}$) of particles. Its Hilbert space can be taken to be finite-dimensional because we can take the systems we consider to be enclosed in a finite volume of 3-space, which leads to a discrete spectrum of the Hamiltonian and thus only finitely many eigenvalues in every finite energy interval. As a consequence, if we cut off high energies at an arbitrary level, or if we consider an energy shell $\Hilbert_{(E-\Delta E,E]}$, we are left with a finite-dimensional Hilbert space of relevant pure states.

\subsection{Macro Spaces}

The orthogonal decomposition $\Hilbert= \oplus_\nu \Hilbert_\nu$ into macro spaces can be thought of as arising from macro observables as follows, following \citet{vN29}. If $\hat A_1, \ldots, \hat A_K$ are quantum observables (self-adjoint operators on $\Hilbert$) that should serve as macro observables, we will want to coarse grain them according to the resolution $\Delta M_j$ of macroscopic measurements by introducing $\hat B_j = [\hat A_j/\Delta M_j]\Delta M_j$ (with $[\cdot]$ again the nearest integer), thereby grouping nearby eigenvalues of $\hat A_j$ together into a single, highly degenerate eigenvalue of $\hat B_j$. Due to their macroscopic character, the commutators $[\hat B_j, \hat B_k]$ will be small; as \citet{vN29} argued and \citet{Og} justified, a little correction will make them commute exactly, $\hat M_j \approx \hat B_j$ and $[\hat M_j, \hat M_k]=0$. Since a macro description $\nu$ of the system can be given by specifying the eigenvalues $(m_1,\ldots, m_K)$ of all $\hat M_1,\ldots,\hat M_K$, their joint eigenspaces are the macro spaces $\Hilbert_\nu$. They have dimensions of order $10^{10^{23}}$.

\subsection{Some Historical Approaches}
\label{sec:historical}

\citet{Pau} tried to derive entropy increase from quantum mechanics. In fact, he described several approaches. In \S 1 of his paper, he developed a modified Boltzmann equation that takes Fermi-Dirac and Bose-Einstein statistics into account and makes sense in the individualist view; in \S 2 he took an ensemblist attitude and suggested a novel dynamics in the form of a Markov chain on the macro states $\nu$, for which he proved that the (analog of the) Gibbs entropy of the probability distribution increases with time; in \S 3 and 4, he argued (not convincingly, to our minds) that the distribution over $\nu$'s provided by an evolving quantum state is, to a good degree of approximation, a Markov chain. 

Returning to Pauli's \S 1, we note that his derivation of the modified Boltzmann equation is far less complete than Boltzmann's derivation of the Boltzmann equation from classical dynamics of hard spheres, and no validity theorem analogous to Lanford's has been proven for it, as far as we are aware. So Pauli's description has rather the status of a conjecture. According to it, the macro evolution $\nu(t)$ of a quantum gas with $N$ particles is deterministic for large $N$ (as with the Boltzmann equation), the expression for the entropy of $\nu$ agrees with the quantum Boltzmann entropy $S_{\qB}(\nu)$ as in \eqref{SqBdef}, and it increases with time as a mathematical consequence of the evolution equation for $\nu(t)$. This approach fits into the framework of an orthogonal decomposition $\Hilbert=\oplus_\nu \Hilbert_\nu$ that arises in a way analogous to Boltzmann's example \eqref{occupation} from an orthogonal decomposition of the 1-particle Hilbert space, $\Hilbert_1=\oplus_i \Kilbert_i$, by taking $\nu$ to be the list of occupation numbers (within tolerances) of each $\Kilbert_i$ and $\Hilbert_\nu$ the corresponding subspaces of $\Hilbert$.

\bigskip

Von Neumann famously came up with the von Neumann entropy \eqref{SvN}, the quantum analog of the Gibbs entropy. It is less known that in 1929 he defended a different formula because \citep[Sec.~1.3]{vN29}
\begin{quotation}
[t]he expressions for entropy [of the form $-\kk \tr \hat\rho \log \hat\rho$]
are not applicable here in the way they were 
intended, as they were computed from the perspective of an observer who can carry out all measurements that are possible in principle---i.e., regardless of whether they are macroscopic[.] 
\end{quotation}
We see here a mix of ensemblist and individualist thoughts in von Neumann's words. Be that as it may, the formula that von Neumann defended was, in our notation,
\be\label{SvN2}
S(\psi)=
-\kk\sum_{\nu} \|\hat P_\nu \psi\|^2 \log \frac{\|\hat P_\nu \psi\|^2 }{\dim \Hilbert_\nu}\,.
\ee
Similar expressions were advocated recently by Safranek et al.~\citeyearpar{SDA17,SDA18}. Note that this expression has a stronger individualist character than $-\kk \tr \hat\rho \log \hat\rho$, as it ascribes a non-trivial entropy value also to a system in a pure state. The problem with the expression is that it tends to average different entropy values where no averaging is appropriate. For example, suppose that $\psi$ has contributions from many but not too many macro spaces; since $\dim \Hilbert_\nu$ is very large, $\|\hat P_\nu \psi\|^2 \log \|\hat P_\nu \psi\|^2$ can be neglected in comparison, and then \eqref{SvN2} is just the weighted average of the quantum Boltzmann entropies \eqref{SqBdef} of different macro states, weighted with $\|\hat P_\nu \psi\|^2$. But in many situations, such as Schr\"odinger's cat, we would not be inclined to regard this average value as the true entropy value valid in all cases. For comparison, we would also not be inclined to define the body temperature of Schr\"odinger's cat as the average over all contributions to the wave function---we would say that the quantum state is a superposition of states with significantly different body temperatures. Likewise, we should say that the quantum state is a superposition of states with significantly different entropies.

\bigskip

In the same paper, \citet{vN29} proved two theorems that he called the ``quantum $H$ theorem'' and the ``quantum ergodic theorem.'' However, they are not at all quantum analogs of either the $H$ theorem or the ergodic theorem. In fact, the two 1929 theorems are very similar to each other (almost reformulations of each other), and they assert that for a typical Hamiltonian or a typical choice of the decomposition $\Hilbert=\oplus_\nu \Hilbert_\nu$, all initial pure states $\psi_0$ evolve such that for most $t$ in the long run, 
\be
\|\hat P_\nu \psi(t)\|^2 \approx \frac{\dim \Hilbert_\nu}{\dim \Hilbert}\,.
\ee
(See also \citet{GLTZ10}.) This is an interesting observation about the quantum evolution that has no classical analog; but it is far from establishing the increase of entropy in quantum mechanics.

\subsection{Entanglement Entropy}
\label{sec:entangle}

Let us come back to option (d) of Section~\ref{sec:SQM} above: the possibility that the density matrix $\hat\rho$ in $S_\vN$ should be the system's reduced density matrix. This approach has the undesirable consequence that entropy is not extensive, i.e., the entropy of $\sys_1 \cup \sys_2$ is not the sum of the entropy of $\sys_1$ and the entropy of $\sys_2$; after all, often the entropies of $\sys_1$ and $\sys_2$ will each increase while that of $\sys_1 \cup \sys_2$ is constant if that is a closed system (because the evolution is unitary). Moreover, it is not true of \emph{all} $\hat\rho$ on $\Hilbert_1\otimes \Hilbert_2$ that $S_\vN(\tr_2 \hat\rho)$ increases with time (think of time reversal).

A variant of option (d), and another widespread approach to defining the entropy of a macroscopic quantum system $\sys$ with (say) a pure state $\Psi$, is to use the entanglement entropy as follows: Divide $\sys$ into small subsystems $\sys_1,\ldots, \sys_M$, let $\hat \rho_j$ be the reduced density matrix of $\sys_j$,
\be\label{rhojdef}
\hat \rho_j := \tr_1 \cdots \tr_{j-1} \tr_{j+1} \cdots \tr_M |\Psi\rangle \langle \Psi|\,,
\ee
and add up the von Neumann entropies of $\hat\rho_j$,
\be\label{Sentdef}
S_\mathrm{ent}:=-\kk\sum_{j=1}^M \tr (\hat\rho_j \log \hat \rho_j)\,. 
\ee
Since the $j$-th term quantifies how strongly $\sys_j$ is entangled with the rest of $\sys$, it is known as the entanglement entropy of $\sys_j$. We will simply call also their sum \eqref{Sentdef} the entanglement entropy.

In thermal equilibrium, the entanglement entropy $S_\mathrm{ent}$ actually agrees with the thermodynamic entropy. That is a consequence of microscopic thermal equilibrium as described in Section~\ref{sec:thermalization}, and thus of canonical typicality \citep{GMM04,GLTZ06,PSW06}: For most $\Psi$ in an energy shell of $\sys$, assuming that the $\sys_j$ are not too large \citep{GHLT17}, $\hat\rho_j$ is approximately canonical, so its entanglement entropy agrees with its thermodynamic entropy, and extensivity implies the same for $\sys$.

This argument also shows that for $\sys$ in thermal equilibrium, $S_\mathrm{ent}$ does not depend much on how $\sys$ is split into subsystems $\sys_j$. On top of that, it also shows that $S_\mathrm{ent}$ still yields the correct value of entropy when $\sys$ is in \emph{local} thermal equilibrium, provided each $\sys_j$ is so small as to be approximately in thermal equilibrium. 

Many practical examples of non-equilibrium systems are in local thermal equilibrium. Still, for general non-equilibrium systems, $S_\mathrm{ent}$ may yield wrong values; an artificial example is provided by a product state $\Psi = \otimes_j \psi_j$ with each $\psi_j$ an equilibrium state of $\sys_j$, so $S_\mathrm{ent}=0$ while the thermodynamic entropy, by extensivity, should have its equilibrium value. Artificial or not, this example at least forces us to say that $S_\mathrm{ent}$ is not correct for \emph{all} $\Psi$; that $S_\mathrm{ent}$ is correct for \emph{most} $\Psi$ we knew already since most $\Psi$ are in thermal equilibrium.

Another issue with $S_\mathrm{ent}$ arises when $\sys$ is in a macroscopic superposition such as Schr\"odinger's cat. Then $S_\mathrm{ent}$ yields approximately the average of the entropy of a live cat and that of a dead cat, while the correct value would be \emph{either} that of a live cat \emph{or} that of a dead cat. In this situation one may want to conditionalize on the case of (say) the cat being alive, and replace $\Psi$ by the part of it corresponding to a live cat. Ultimately, this leads us to replace $\Psi$ by its projection to one $\Hilbert_\nu$, and if we are so lucky that states in $\Hilbert_\nu$ are in local thermal equilibrium, then the conditionalized $S_\mathrm{ent}$ will agree with $S_{\qB}(\nu)$. So at the end of the day, $S_{\qB}$ seems to be our best choice as the fundamental definition of entropy in quantum mechanics.

\subsection{Increase of Quantum Boltzmann Entropy}
\label{sec:SqBincrease}

Just as the classical macro sets $\Gamma_\nu$ have vastly different volumes, the quantum macro spaces $\Hilbert_\nu$ have vastly different dimensions, and as phase space volume is conserved classically, dimensions of subspaces are conserved under unitary evolution. 

It follows that if macro states $\nu$ follow an autonomous, deterministic evolution law $\nu\mapsto\nu_t$, then quantum Boltzmann entropy increases. In fact, 
\be
\text{if }e^{-i\hat H t}\Hilbert_\nu \subseteq \Hilbert_{\nu'},
\text{ then }S_\qB(\nu') \geq S_\qB(\nu).
\ee

More generally, not assuming that macro states evolve deterministically, already at this point it is perhaps not unreasonable to expect that in some sense the unitary evolution will carry a pure state $\psi_t$ from smaller to larger subspaces $\Hilbert_\nu$. In fact, it was proven \citep{GLMTZ10} that if one subspace has most of the dimensions in an energy shell, then $\psi_t$ will sooner or later come very close to that subspace (i.e., reach thermal equilibrium) and stay close to it for an extraordinarily long time (i.e., for all practical purposes, never leave it).

Proposition~\ref{prop:2times} below is another observation in the direction of Conjecture~\ref{conj:qB2ndlaw}, concerning only two times instead of an entire history and only two macro states $\nu,\nu'$. It is a quantum analog of the classical statement that if 
\be
\vol \Gamma_{\nu'} \ll \vol \Gamma_{\nu}\,,
\ee
that is, if $S_\B(\nu')$ is less than $S_\B(\nu)$ by an appreciable amount, then for any $t\neq 0$, most $X_0\in \Gamma_\nu$ are such that $X_t\notin \Gamma_{\nu'}$.

\begin{prop}\label{prop:2times}
Let $\hat H$ and $t\neq 0$ be arbitrary but fixed. If
\be
\dim \Hilbert_{\nu'} \ll \dim \Hilbert_{\nu}
\ee
(that is, if $S_\qB(\nu')$ is less than $S_\qB(\nu)$ by an appreciable amount), then for most $\psi_0\in \SSS(\Hilbert_\nu)$,
\be\label{2times}
\bigl\| \hat P_{\nu'} \psi_t\bigr\|^2 \ll 1\,.
\ee
\end{prop}

\begin{proof}
Let $\hat U = e^{-i\hat Ht}$, let $d_\nu:= \dim \Hilbert_\nu$, and let $\Psi$ be a random vector with uniform distribution on $\SSS(\Hilbert_\nu)$. Then
\be\label{uniformrho}
\EEE \pr{\Psi}=d_{\nu}^{-1}\hat P_\nu\,,
\ee
so
\begin{align}
\EEE \| \hat P_{\nu'}\hat U\Psi\|^2 &=\tr(\hat P_{\nu'} \hat U d_\nu^{-1} \hat P_{\nu} \hat U^{-1})\\
&=d_\nu^{-1} \tr(\hat U^{-1} \hat P_{\nu'} \hat U \hat P_{\nu})\\
&\leq d_\nu^{-1} \tr(\hat U^{-1} \hat P_{\nu'} \hat U )\\
&= d_\nu^{-1} \tr( \hat P_{\nu'})\\
&= \frac{\dim \Hilbert_{\nu'}}{\dim \Hilbert_\nu}\ll 1\,.
\end{align}
If a non-negative quantity is small on average, it must be small in most cases, so the proposition follows.
\end{proof}

Note that the proof only used that $\Hilbert_{\nu'}$ and $\Hilbert_{\nu}$ are two mutually orthogonal subspaces of suitable dimensions; as a consequence, the statement remains true if we replace $\Hilbert_{\nu'}$ by the sum of all macro spaces with dimension less than that of $\Hilbert_\nu$. That is, if $\Hilbert_{\nu}$ is much bigger than all of the smaller macro spaces together (a case that is not unrealistic), then it is atypical for entropy to be lower at time $t$ than at time 0.

\section{Conclusions}
\label{sec:conclusions}

We have argued that entropy has nothing to do with the knowledge of observers. The Gibbs entropy \eqref{SGdef} is an efficient tool for computing entropy values in thermal equilibrium when applied to the Gibbsian equilibrium ensembles $\rho$, but the fundamental definition of entropy is the Boltzmann entropy \eqref{SBdef}. We have discussed the status of the two notions of entropy and of the corresponding two notions of thermal equilibrium, the ``ensemblist'' and the ``individualist'' view. Gibbs's ensembles are very useful, in particular as they allow the efficient computation of thermodynamic functions \citep{Gol19}, but their role can only be understood in Boltzmann's individualist framework. We have also outlined an extension of Boltzmann's framework to quantum mechanics by formulating a definition of the quantum Boltzmann entropy and, as Conjecture~\ref{conj:qB2ndlaw}, a statement of the second law of thermodynamics for it.

\bigskip

\noindent{\bf Acknowledgments.} We thank Jean Bricmont and David Huse for valuable discussions. The work of JLL was supported by the AFOSR under award number FA9500-16-1-0037. JLL thanks Stanislas Leibler for 
hospitality at the IAS.

\end{document}